\documentclass[journal]{IEEEtran}

\usepackage{amsmath}
\usepackage{tabularx}
\usepackage{graphicx}
\usepackage{subfigure}
\usepackage{amssymb}
\usepackage{amsthm}
\usepackage{makecell}
\usepackage{diagbox}
\usepackage{algorithm,algpseudocode}
\newtheorem{theorem}{Theorem}
\newtheorem{corollary}{Corollary}
\newtheorem{lemma}{Lemma}

\usepackage{url} 
\usepackage{rotating}

\begin{document}

\title{A Stackelberg Game Approach for Two-Level Distributed Energy Management in Smart Grids}

\author{   Juntao~Chen,~\IEEEmembership{Student Member,~IEEE,}
       and Quanyan~Zhu,~\IEEEmembership{Member,~IEEE}

\thanks{This paper has been accepted to be published in  \textit{IEEE Transactions on Smart Grids}.}

\thanks{The authors are with the Department of Electrical and Computer Engineering, Tandon School of Engineering, New York University, Brooklyn, NY 11201 USA. E-mail: \{jc6412, qz494\}@nyu.edu.}

\thanks{This research is partially supported by a DHS grant through Critical Infrastructure Resilience Institute (CIRI) and grants EFRI-1441140, SES-1541164, and ECCS-1550000 from National Science Foundation (NSF).}
}

\maketitle

\begin{abstract}
The pursuit of sustainability motivates microgrids that depend on distributed resources to produce more renewable energies. An efficient operation and planning relies on a holistic framework that takes into account the interdependent decision-making of the generators of the existing power grids and the distributed resources of the microgrid in the integrated system. To this end, we use a Stackelberg game-theoretic framework to study the interactions between generators (leaders) and microgrids (followers). Entities on both sides make strategic decisions on the amount of power generation to maximize their payoffs. Our framework not only takes into account the economic factors but also incorporates the stability and efficiency of the smart grid, such as the power flow constraints and voltage angle regulations. We present three update schemes for microgrids. In addition, we develop three other algorithms for generators, and among which a fully distributed algorithm enabled by phasor measurement units is proposed. The distributed algorithm merely requires the information of voltage angles at local buses for updates, and its convergence to the unique equilibrium is shown. We further develop the implementation architectures of the update schemes in the smart grid. Finally, case studies are used to corroborate the effectiveness of the proposed algorithms.
\end{abstract}


\section{Introduction}
In the future smart grid, a large number of green energy systems that depend on renewable distributed resources, such as solar, wind, biomass and geothermal, will be built and integrated with the current main power grids \cite{AEO}.
These distributed resources can be built into smart microgrids, which are green energy systems that can operate independently for self-efficiency. For example, in remote and wild areas, wind farm consisting of wind turbines is a possible method to satisfy its local power demands \cite{chen2015resilient,dietrich2012demand,chen2015optimal}. In addition, smart microgrids can be connected with the main power system which is beneficial for grid dependability and resiliency \cite{farzinenhancing}. In this case, the microgrids can enter the power market to sell renewable energies or buy electricity from the external grid. The decision of each microgrid is based on its objective and thus is strategic. Comparing with microgrids, generators in the power system are generally equipped with a larger capacity, and they play dominant roles in the grid and determine the electricity price in the power market. By considering the power balance and stability issues of the system, generators' strategy on the power generation will influence the microgrids' decision, which introduces interdependencies between these two separate entities in the energy system.

In this paper, we study the interactions and decision interdependencies between multiple generators and microgrids in the smart grid. Specifically, the generators aim to maximize their revenues by determining the power generations supplied to the loads in the system. Based on the generators' decision, each microgrid makes a strategy on the amount of renewable energy injection to the grid and serving for its local demand. The inclusion of microgrid entities in the power system leads to a competing mechanism among themselves. In addition, the strategies of microgrids will in turn impact the decision of the generators. A general two-layer framework of smart grid considered is depicted in Fig. \ref{layer1}, in which the upper layer includes generators and the bottom layer contains microgrids. Furthermore, the communication network between two layers is used to exchange information, such as the amount of power generation and electricity price. 
\begin{figure}[!t]
\centering
\includegraphics[width=0.8\columnwidth]{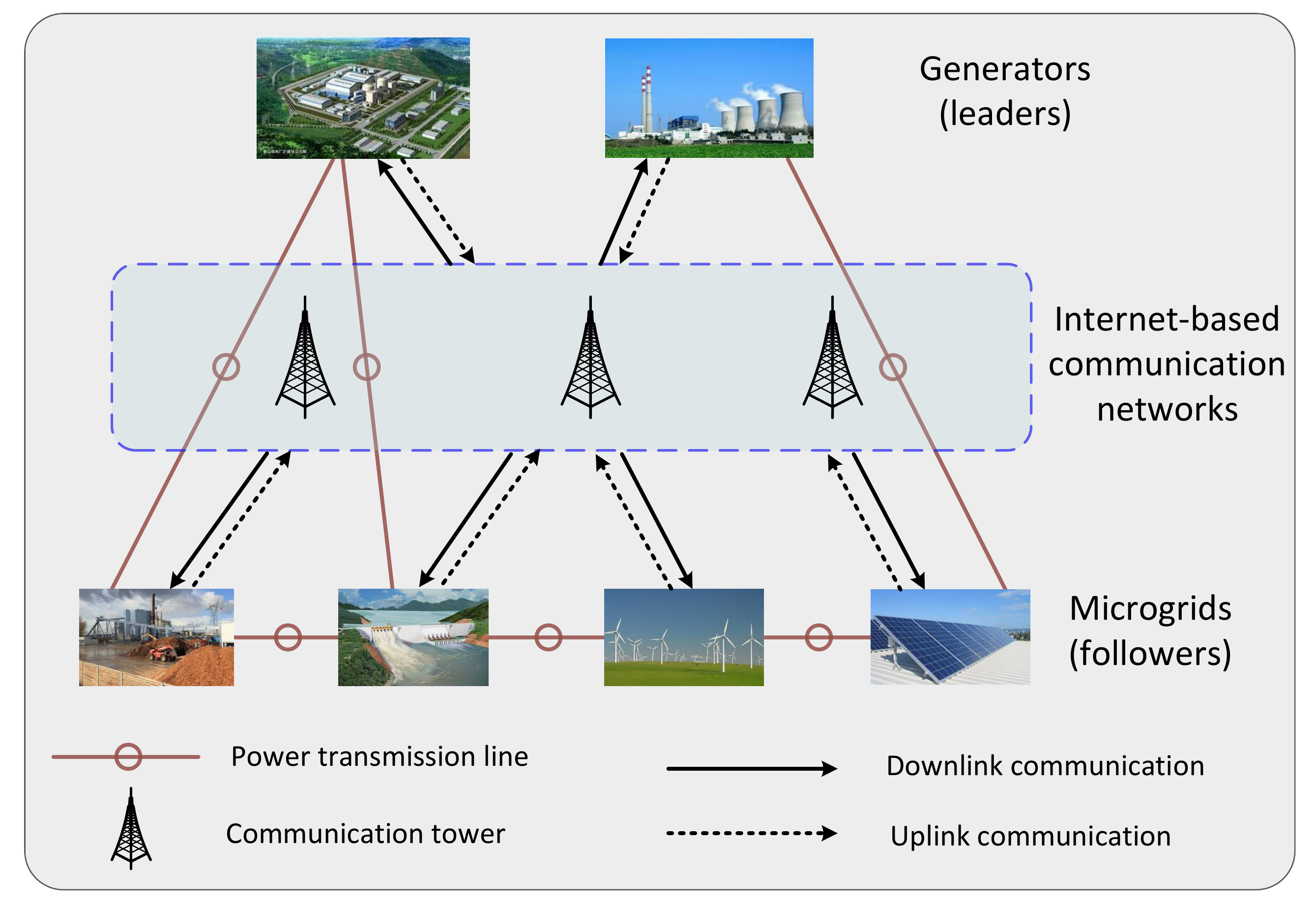}
\caption{Illustration of a general two-layer framework of smart grid which includes a generator network and a microgrid network. Information exchange between two layers is enabled by the communication network.}\label{layer1}
\end{figure}

In our framework, the payoffs of generators and microgrids are related to both physical and economic factors which include the system power flows, renewable generation capacity, generation cost and electricity market price \cite{ li2012automated, zhu2012game}. Since generators play dominant roles in the smart grid and power market, they can be seen as leaders during the decision-making processes, while microgrids are regarded as followers. The complex cross-layer interactions in the smart grid motivate a Stackelberg game framework to understand the interdependent decision-making \cite{basar1995dynamic}. Note that the generators and microgrids can be viewed as players in this game. Specifically, the microgrids at the bottom layer play a non-cooperative Nash game, and they aim to find the best operation and planning decisions in response to the generators' strategy. For the integrated smart grid, the outcome of the strategic decision-making of the multi-layer and multi-agent interactions is characterized by Stackelberg equilibrium which yields the generators' and microgrids' strategies jointly.

To design an advanced smart grid energy management system, we first propose two communication-based update algorithms for generators and microgrids to enable their decision-makings. However, the communication distance between players in the grid can be long which makes the long-range communication either infeasible or not economical. To this end, we further develop a PMU-enabled distributed scheme which merely depends on the phasor measurement unit (PMU) devices  \cite{de2010synchronized} to measure the voltage angles at buses. With this scheme, the generators and microgrids do not need to exchange their private information, such as power generation and voltage angle regulation parameters, and make optimal decisions independently, preserving high confidentiality for the players.

The contributions of this paper are summarized as follows.
\begin{enumerate}
\item We establish a Stackelberg game-theoretic framework to capture the strategic generation plannings of generators and microgrids in the smart grid and characterize the solution via Stackelberg equilibrium (SE).
\item The decision-making of the players considers not only the economic factors but also the physical constraints of the power system in the framework.
\item We develop distributed algorithms to find the equilibrium strategy and derive sufficient conditions to ensure its convergence to a unique SE.
\item To implement the proposed update schemes, we propose a system architecture that builds on PMU and wireless communication infrastructure of the smart grid.
\end{enumerate}

\subsection{Related Work}
Game-theoretic methods have been widely used to capture the strategic behaviors of agents in infrastructure systems, e.g., power systems and communication systems \cite{zhu2012game,mohsenian2010autonomous,
nekouei2015game,zhu2012differential, chen2015resilient3,kurian2017electric,
chen2016optimal,zhu2015game,chen2016interdependent}. In hierarchical power systems, a Stackelberg game approach has been used to model the demand response between users and utility companies \cite{maharjan2013dependable,
yu2016real,zhu2013multi}. However, they have mainly focused on the economic aspect of power generation planning and dispatch, while physical constraints, such as power flows and voltage angle regulations, are not included in the framework. Our work extends their multilevel model to a more practical one in the Stackelberg game setting and design distributed algorithms to enable the decision-making of players. The upper layer's problem in our framework can be categorized into a mathematical program with equilibrium constraints (MPEC) \cite{luo1996mathematical,hobbs2000strategic}, and we obtain analytical results first to solve it in a distributed fashion.

In \cite{chen2017game}, the authors have establised a two-layer framework and addressed the problem of microgrids generation planning in a Nash game by assuming that the generators' decision are given. In the current work, we view the generators at the upper layer as leaders and establish a Stackelberg game framework to deal with a more challenging two-level decisoin-making problem. The presented three update schemes for microgrids in Section \ref{microgrid_algorithms} mainly follow the work \cite{chen2017game}. In this work, we focus more on the generator/leader side analysis in Section \ref{generator_analysis} and the update algorithms design, especially the novel fully distributed one for generators presented in Section  \ref{generator_algorithm}. In addition, we design implementation architectures for the players at both layers in Section \ref{architecture} which differs from \cite{chen2017game}.

\subsection{Organization of the Paper}
The remainder of the paper is organized as follows. Section \ref{pre} introduces the preliminaries of power flows. In Section \ref{game_framework}, we formulate a Stackelberg game-theoretic framework to capture the strategic behaviors of microgrids and generators, and define its solution concept. Game analysis is presented in Section \ref{gameanalysis}, and the update schemes to find the equilibrium solution are developed in Section \ref{algorithms}. Implementation frameworks of the algorithms in smart grid are designed in Section \ref{architecture}. Case studies are given in Section \ref{cases}, and Section \ref{conclusion} concludes the paper.

\subsection{Notations and Conventions}
Some notations and conventions are summarized in Table \ref{table1}. Superscript $*$ indicates the equilibrium. In addition, subscripts $l$ and $f$ represent the leader and the follower, respectively. Microgrid and follower refer to the same entity and used interchangeably, and similar for the terms of ``generator'' and ``leader''.

\begin{table}[t]
\centering
\renewcommand\arraystretch{1.4}
\caption{Nomenclature\label{table1}}
\begin{tabular}{ll} \hline
$P_i^g$ & power generation at bus $i$\\ 
$P_{i,\max}^g$ & maximum power generation at bus $i$\\ 
$P_i^l$ & power load at bus $i$\\ 
$P_i$ & active power injection at bus $i$\\ 
$\theta_i$ & voltage angle at bus $i$\\ 
$N_d$ & number of microgrids\\
$N_g$ & number of generators\\
$\psi_i$ & unit generation cost of microgrid $i$\\
$\zeta$ & unit market power price \\
$\eta_i$ & weighting constant of voltage angle regulations at microgrid $i$ \\
$\alpha_j$ & weighting constant of voltage angle regulations at generator $j$ \\
$a_j,b_j,c_j$ & parameters in generator $j$'s generation cost function\\
$J_{f,i}$ & objective of microgrid $i$\\
$J_{l}$ & objective of generators\\
$\mathbf{P}$ & $\mathbf{P}:=[P_1,P_2,...,P_N]^T\in\mathbb{R}^N$ \\
$\mathbf{P}_d$ & $\mathbf{P}_d:=[P_1,P_2,...,P_{N_d}]^T\in\mathbb{R}^{N_d}$\\
$\mathbf{P}_g$ & $\mathbf{P}_g:=[P_{N_d+1},P_{N_d+2},...,P_{N}]^T\in\mathbb{R}^{N_g}$\\ 
$\mathbf{P}_d^g$ & $\mathbf{P}_d^g:=[P_1^g,P_2^g,...,P_{N_d}^g]^T\in\mathbb{R}^{N_d}$ \\
$\mathbf{P}_g^g$ & $\mathbf{P}_g^g:=[P_{N_d+1}^g,P_{N_d+2}^g,...,P_{N}^g]^T\in\mathbb{R}^{N_g}$\\
$\mathbf{\theta}$ & $\mathbf{\theta}:=[\theta_1,\theta_2,...,\theta_N]^T\in\mathbb{R}^N$\\ 
$\boldsymbol{\theta}_d $ & $\boldsymbol{\theta}_d := [\theta_1,\theta_2,...,\theta_{N_d}]^T\in\mathbb{R}^{N_d}$ \\
$\boldsymbol{\theta}_g$ & $\boldsymbol{\theta}_g := [\theta_{N_d+1},\theta_{N_d+2},...,\theta_{N}]^T\in\mathbb{R}^{N_g}$ \\
$\mathcal{N}_d$ & set of buses connected with microgrids\\
$\mathcal{N}_g$ & set of buses connected with generators\\
$\mathcal{P}$ & feasible set of active power injection of all buses\\
$\mathcal{P}^{\mathcal{G}}_{f,i}$ & set of microgrid $i$'s generation \\
$\Theta_{i}$ & set of microgrid $i$'s voltage angle\\
$\mathcal{P}^{\mathcal{G}}_{f,-i}$ &  Cartesian product of $\mathcal{P}^{\mathcal{G}}_{f,j},\ \forall j\in\mathcal{N}_d\setminus\{i\}$ \\
$\Theta_{-i}$ &  Cartesian product of $\Theta_{j},\ \forall j\in \mathcal{N}_d\setminus\{i\}$  \\
$\mathcal{P}_{f,F}^{\mathcal{G}}$ & feasible generation set of game $G_f$\\
$\Theta_{f,F}$ & feasible voltage angle profile of game $G_f$\\ 
 $\mathcal{F}_f$ & feasible action set of game $G_f$ \\
 $ \mathcal{P}^{\mathcal{G}}_{l,j}$ & set of generator $j$'s generation \\
 $\mathcal{P}_{l,F}^{\mathcal{G}}$ & feasible generation set of generators\\
 $\Theta_{l,F}$ & feasible voltage angle set of generators\\
 $\mathcal{F}_l$ & feasible action set of generators\\
 \hline
\end{tabular}
\end{table}

\section{Power Flow Preliminaries}\label{pre}
In a smart power grid consisting of $N+1$ buses, let $\mathcal{N}:=\{r,1,2,...,N\}$, and $r$ is the slack bus. In addition, denote $P_i$, $Q_i$, $V_i$ and $\theta_i$ as the amount of active power injection, reactive power injection, voltage magnitude and voltage angle at bus $i$, $i\in\mathcal{N}$, respectively. Then, the power flow equations of the system with reference to the slack bus $r$ are given by
\begin{equation}\label{PQ}
\begin{split}
P_i=\sum\limits_{j\in\mathcal{N}} V_iV_j[G_{ij}\cos(\theta_i-\theta_j)+B_{ij}\sin(\theta_i-\theta_j)],\\
Q_i=\sum\limits_{j\in\mathcal{N}} V_iV_j[G_{ij}\sin(\theta_i-\theta_j)-B_{ij}\cos(\theta_i-\theta_j)],
\end{split}
\end{equation}
for $i,j=1,2,...,N,$ where $G_{ij}$ and $B_{ij}$ represent the real part and imaginary part of element $(i,j)$ in the admittance matrix $\mathbf{Y}\in\mathbb{C}^{N\times N}$ of the power system. $V_r$ and $\theta_r$ are both known, and $\theta_r=0$ by default.

Furthermore, let $P_i^g$ and $P_i^l$ be the power generation and power load at bus $i$, respectively. Then, the active power injection at bus $i$ has the relation
\begin{equation}\label{injection}
P_i=P_i^g-P_i^l,\ \forall i\in\mathcal{N}.
\end{equation}
Moreover, by considering the balance of the power grid, we have
$\sum_{i\in\mathcal{N}} P_i^g=\sum_{i\in\mathcal{N}} P_i^l.$

For notational clarity in the system representation and analysis, the set $\mathcal{N}$ excludes the slack bus notation $r$ in the following. In power system analysis, DC approximation is used for fast calculation of power flows \cite{glover2011power}.  We assume that the resistance of transmission lines is much smaller than its reactance; the voltage angles $\theta_i,\ \forall i\in \mathcal{N}$, are small, and the magnitudes of voltages $V_i,\ \forall i\in\mathcal{N}$, are equal to 1 p.u. Then, $Q_i=0,\ \forall i\in\mathcal{N}$, $G_{ij}\ll B_{ij}$, $\sin(\theta_i-\theta_j)\approx \theta_i-\theta_j$ and $\cos(\theta_i-\theta_j)\approx 1$. Therefore, power flow equations \eqref{PQ} can be represented by a set of linear equations
$
P_i=\sum_{j\neq i}B_{ij}(\theta_i-\theta_j),\ \forall i,j\in\mathcal{N},
$
which can be rewritten in following matrix form
\begin{equation}\label{power_dc}
\mathrm{\mathbf{P}}=-\mathrm{\mathbf{B}}\boldsymbol{\theta},
\end{equation}
where $\mathbf{P}=[P_1,P_2,...,P_N]^T\in\mathbb{R}^N$ and $\mathbf{\theta}=[\theta_1,\theta_2,...,\theta_N]^T\in\mathbb{R}^N$.

Matrix $\mathbf{B}$ contains the imaginary components of $\mathbf{Y}$ except the slack bus's row and column. Note that $-\mathbf{B}$ is a symmetric reduced Laplacian matrix, and hence $\mathbf{B}$ is invertible through the Kirchhoff's matrix-tree theorem \cite{chaiken1982combinatorial}.
Since $\mathbf{B}$ is nonsingular, \eqref{power_dc} can be rewritten as 
\begin{equation}\label{theta}
\boldsymbol\theta=\mathbf{SP},
\end{equation} 
where $\mathbf{S}:=[s_{ij}]_{i,j\in \mathcal{N}}=-\mathbf{B}^{-1}$.

\begin{lemma}[Lemma 1, \cite{chen2017game}]\label{lemma_Smatrix}
Matrix $\mathbf{S}$ is symmetric  and ${s_{ij}\geq 0},\ \forall\ i,j\in\mathcal{N}$, and especially $s_{ii}>0,\ \forall\ i\in\mathcal{N}$.
 \end{lemma}
Lemma \ref{lemma_Smatrix} is useful for the analysis in Section \ref{gameanalysis}.

\section{Stackelberg Game-Theoretic Framework}\label{game_framework}
In this section, we present the microgrid side and generator side game-theoretic models in detail and then define the equilibrium solution concept.

\subsection{System Overview and Couplings}
In a smart energy system including a number of microgrids and generators, we denote $\mathcal{N}_d:=\{1,2,...,N_d\}$ as a set of $N_d$ buses that are connected with the microgrids, and they are able to generate renewable energies such as wind and solar power. In addition, denote $\mathcal{N}_g:=\{N_d+1,N_d+2,...,N\}=\mathcal{N}\setminus\mathcal{N}_d$ as a set of $N_g$ buses that are generators. The slack bus $r$ is further chosen to serve as a base of the power grid. For bus $i\in\mathcal{N}$ having local loads, its value $P_i^l$ is specified ahead of time. Both generators and microgrids in the system need to determine their amount of power generation, and they have a maximum generation limit
$0\leq P^g_i \leq P^g_{i,\max},\ \forall i\in\mathcal{N}.$ Note that microgrid's generation capacity $P^g_{i,\max},\ i\in\mathcal{N}_d$ is dynamically changing due to the intermittent nature of renewable generations. For convenience, the framework is studied for a given $P^g_{i,\max}$ in this paper, and it can be generalized to the dynamic case naturally.
Without loss of generality, we set the power loads at the generator buses to 0, i.e., $P_i^l = 0,\ \forall i\in\mathcal{N}_g$, then $\mathbf{P}_g=\mathbf{P}_g^{g}$.

As the framework shown in Fig. \ref{layer1}, generators at the upper layer are viewed as a leader, and they have the priority of determining the amount of power generations supplied to the loads at the lower layer. Microgrids are followers in this generation game, and their decisions need to take the strategies of the upper layer into account. In particular, the interactions between microgrids constitute a Nash game, and the couplings between two layers are captured by a Stackelberg game-theoretic framework. 

\subsection{Microgrid (Follower) Side Model}
Before formulating the Nash game between microgrids, we present underlying assumptions of the model. First, the topology of the whole power system is known to all players. This is justifiable since the parameters of power transmission lines are often known. Second, the constraints of each player are common information, and each player is aware of the physical constraints when making decisions. This indicates that players should take the power flow constraints \eqref{power_dc} into account. Third, during microgrids' updates, the power generations of generators are fixed. This assumption is reasonable since microgrids can regulate themselves more quickly than generators, and they can be seen as followers in this game \cite{zhu2012game,maharjan2013dependable,chai2014demand}. Furthermore, PMUs are installed at buses to measure the voltage angles in the smart grid. The game played among the microgrids are presented as follows.

Let  $G_f:=\{\mathcal{N}_d,{\{ \mathcal{P}^{\mathcal{G}}_{f,i},\Theta_i\} }_{i\in \mathcal{N}_d},{\{ J_{f,i} \}}_{i\in \mathcal{N}_d}, \mathcal{P} \}$ be a game with a set $\mathcal{N}_d$ of $N_d$ microgrids. Here, we use subscript $f$ to denote the \textit{follower} for convenience. ${\{ \mathcal{P}^{\mathcal{G}}_{f,i},\Theta_i\} }$ is the action set of microgrid $i$, where
$\mathcal{P}^{\mathcal{G}}_{f,i}:=\{ P_i^g\in\mathbb{R}_+\ |\ 0\leq P_i^g \leq P^g_{i,\max} \};$ Denote ${\{ \mathcal{P}^{\mathcal{G}}_{f,-i},\Theta_{-i}\} }$ as Cartesian product of all microgrids' action sets except $i$th one.
$\mathcal{P}$ is the feasible set of active power injections defined by constraint \eqref{power_dc}. Denote the feasible set of power generations of all buses in the grid as $\mathcal{P}^{\mathcal{G}}$ which can be obtained by using $\mathcal{P}$ through \eqref{injection}. Then, the feasible generation set of  game $G_f$ can be defined by $\mathcal{P}_{f,F}^{\mathcal{G}}:=\big( \otimes_{i\in \mathcal{N}_d} \mathcal{P}^{\mathcal{G}}_{f,i} \big)\cap\mathcal{P}^{\mathcal{G}}$. The feasible voltage angle profile $\Theta_{f,F}$ of the followers can be obtained based on $\mathcal{P}_{f,F}^{\mathcal{G}}$ and \eqref{power_dc}. Then, the feasible action set of microgrids is $\mathcal{F}_f:=\mathcal{P}_{f,F}^{\mathcal{G}}\times \Theta_{f,F}$. Denote $(P_{f,-i}^{g},\theta_{f,-i})$ by the actions of all microgrids except the $i$'s one for convenience. The followers' game $G_f$ is coupled through the power flow constraints \eqref{power_dc}.

The cost function $J_{f,i}:\mathcal{P}_{f,i}^{\mathcal{G}}\times [0,\pi]\rightarrow \mathbb{R}$ for microgrid $i$ is given by 
\begin{equation}\label{utility}
J_{f,i}(P_i^g,\theta_i)=\psi_i P_i^g+\zeta(P_i^l-P_i^g)+\frac{1}{2}\eta_i^2\theta_i^2,\ \ i\in \mathcal{N}_d,
\end{equation}
where $\psi_i$ is the unit cost of generated power for microgrid $i$; $\zeta$ is the unit price of renewable energy for sale defined by the power market, and $\eta_i$ is a weighting parameter that indicates the importance of regulations of voltage angle at bus $i$. Note that term $\psi_i P_i^g+\zeta(P_i^l-P_i^g)$ captures the renewable generation costs and possible sale revenue of microgrid $i$. In addition, term $\frac{1}{2}\eta_i^2\theta_i^2$ indicates that smaller voltage phasor $\theta_i$ is preferred to improve the power quality in the grid.

Therefore, for a given set $\mathbf{P}_g^g:=[P_{N_d+1}^g,P_{N_d+2}^g,...,P_{N}^g]^T$ of power generations from the generator network, the follower (microgrid) $i$'s problem $(\mathrm{FP}_i)$ is given by
\begin{align*}
\mathrm{FP}_i: \quad\min\limits_{P_i^g,\theta_i}\quad &J_{f,i}(P_i^g,\theta_i)\\
\mathrm{s.t.}\quad &\mathrm{\mathbf{P}}=-\mathbf{B}\boldsymbol{\theta},\\
&0\leq P_i^g \leq P_{i,\max}^g,\ i\in\mathcal{N}_d.
\end{align*}
Note that the generators' output $\mathbf{P}_g$ in $\mathbf{P}$ presents in every $\mathrm{FP}_i$.

\subsection{Generator (Leader) Side Model}
For the generators in the smart grid, we denote $ \mathcal{P}^{\mathcal{G}}_{l,j}$ as the action set of player $j$, where
$\mathcal{P}^{\mathcal{G}}_{l,j}:=\{ P_j^g\in\mathbb{R}_+\ |\ 0\leq P_j^g \leq P^g_{j,\max} \},\ j\in\mathcal{N}_g.$
For convenience, the subscript $l$ stands for \textit{leader}. In addition, denote the feasible generation and voltage angle sets of generators as $\mathcal{P}_{l,F}^{\mathcal{G}}$ and $\Theta_{l,F}$, respectively, and $\mathcal{F}_l:=\mathcal{P}_{l,F}^{\mathcal{G}}\times \Theta_{l,F}$. Denote $(P_{l,-j}^{g},\theta_{l,-j})$ by the actions of all generators except the $j$'s one. Remind that the feasible action sets of microgrids and generators are coupled through \eqref{power_dc}. Each generator $j\in\mathcal{N}_g$ generates an amount of $P_j^g$ power at a cost $C_j(P_j^g)$, where $C_j:\mathcal{P}_{l,j}^{\mathcal{G}} \rightarrow \mathbb{R_+}$ is the cost function for generator $j$. Moreover, we assume that $C_j$ are smooth and convex functions, $\forall j\in\mathcal{N}_g$. A typical choice of $C_j$ is of quadratic form, i.e.,
$
C_j(P_j^g)=\frac{1}{2}a_j {(P_j^g)}^2+b_jP_j^g+c_j,\ j\in\mathcal{N}_g,$\cite{kirschen2004fundamentals,wood2012power},
where $a_j>0,\ b_j\geq 0$, and $c_j\geq 0$ are cost coefficients for generator $j$. The cost function for generator $j$ is 
$J_{l,j}(P_j^g,\theta_j)=C_j(P_j^g)+\frac{1}{2}\alpha_j\theta_j^2,\ j\in\mathcal{N}_g,$
where the second term captures the power quality issue. Specifically, smaller voltage phasor indicates higher power supply efficiency and quality. For generators, the global cost is represented by
$J_{l}(\mathbf{P}_g^g,\boldsymbol{\theta}_g)=\sum_{j\in\mathcal{N}_g}C_j(P_j^g)+\frac{1}{2}\alpha_j\theta_j^2$, which jointly takes the generation costs and power quality into account.
Therefore, the leader's optimization problem (OP\footnote{To avoid ambiguity, we use the abbreviation OP instead of LP to denote the leader's optimization problem.}) is formulated as
\begin{align*}
\mathrm{OP}: \quad\min\limits_{P_j^g,\theta_j}\quad &\sum_{j\in\mathcal{N}_g}C_j(P_j^g)+\frac{1}{2}\alpha_j\theta_j^2\\
\mathrm{s.t.}\quad &\mathrm{\mathbf{P}}=-\mathbf{B}\boldsymbol{\theta},\\
&0\leq P_j^g \leq P_{j,\max}^g,\ \forall j\in\mathcal{N}_g,\\
& \mathrm{Outcome\ strategy\ of}\ \mathrm{FP}_i,\ \forall i\in\mathcal{N}_d.
\end{align*}
OP can be seen as a modified economic power dispatch problem \cite{zhu2012game,zhu2013value,
mohsenian2010autonomous}, since OP includes the constraints from the microgrids layer. 

\textit{Connections between $\mathrm{OP}$ and $\mathrm{FP}_i$:} Remind that the decisions of microgrids yield by $\mathrm{FP}_i$, $\forall i\in\mathcal{N}_d$, are determined for a given generators' action $\mathbf{P}_g^g$, and these decisions constitute the constraints of the outcome strategy of $\mathrm{FP}_i$, $\forall i\in\mathcal{N}_d$, in the OP. Hence, due to this hierarchical smart grid structure, generators at the upper layer (solve OP) need to anticipate the microgrids' response $\mathbf{P}_d^g$ (solve $\mathrm{FP}_i, \forall i\in\mathcal{N}_d$) at the lower layer when making decisions. 

Note that the outcome strategy profile of microgrids is characterized by a Nash equilibrium. Therefore, the OP can be categorized into a mathematical program with equilibrium constraints (MPEC) \cite{luo1996mathematical,hobbs2000strategic}. The hierarchical generation game including the microgrids and generators constitutes a Stackelberg game denoted by $G$.

\subsection{Stackelberg Equilibrium}
The leaders (generators) determine their power generation and announce it to the followers (microgrids). The equilibrium of the microgrids in a Stackelberg game $G$ is any strategy that constitutes an optimal response to the one adopted and announced by the generators. Formally, the Stackelberg equilibrium (SE) is defined as follows.

\textit{Definition 1 (Stackelberg Equilibrium)}: The strategy set $(\mathbf{P}_d^{g*},\boldsymbol{\theta}_d^*,\mathbf{P}_g^{g*},\boldsymbol{\theta}_g^*)$ constitutes a Stackelberg equilibrium of game $G$ if the following conditions are satisfied:

($i$): $(\mathbf{P}_d^{g*},\boldsymbol{\theta}_d^*)\in\mathcal{F}_f$ is a Nash equilibrium for microgrids, i.e., $\forall i\in \mathcal{N}_d$,
 $$J_{f,i}(P_i^{g*},\theta_{i}^{*})\leq J_{f,i}(P_i^g,\theta_i),\ \forall (P_i^g,\theta_i)\in \Phi_i(P_{f,-i}^{g*},\theta_{f,-i}^{*}),$$
 where $\Phi_i(P_{f,-i}^{g*},\theta_{f,-i}^{*})$ is a projected action set defined by $\Phi_i(P_{f,-i}^{g*},\theta_{f,-i}^{*}):=\{(P_i^g,\theta_i):(P_i^g,\theta_i;P_{f,-i}^{g*},\theta_{f,-i}^{*})\in\mathcal{F}_f\}$.
 
($ii$): $(\mathbf{P}_g^{g*},\boldsymbol{\theta}_g^*)\in\mathcal{F}_l$ satisfies, $\forall i\in \mathcal{N}_g$, 
\begin{align*}
&J_l\left(\mathbf{P}_g^{g*},\boldsymbol{\theta}_g^*; \Gamma(\mathbf{P}_g^{g*},\boldsymbol{\theta}_g^*)\right)\\ 
& \leq J_l\left(P_i^g,P_{l,-i}^{g*},\theta_{i},\theta_{l,-i}^{*} ;\Gamma\left(P_i^g,P_{l,-i}^{g*},\theta_{i},\theta_{l,-i}^{*} \right)\right),
\end{align*}
where $\Gamma(\mathbf{P}_g^{g*},\boldsymbol{\theta}_g^*)$ denotes the optimal response of microgrids given $(\mathbf{P}_g^{g*},\boldsymbol{\theta}_g^*)$. 
For a given pair $(\mathbf{P}_g^g,\boldsymbol{\theta}_g)\in\mathcal{F}_l$ of generators, 
 the optimal responses of all microgrids constitute a Nash equilibrium satisfying condition ($i$).
 

In the following, our goal is to find the equilibrium strategies of generators and microgrids jointly.

\section{Stackelberg Game Analysis}\label{gameanalysis}
In this section, we analyze the strategic behaviors of microgrids and generators, respectively, and also show the existence and uniqueness of Stackelberg equilibrium of game $G$.

\subsection{Microgrid Side Analysis}
Before developing algorithms to find the solution to $\mathrm{FP}_i$, we obtain an analytical solution for microgrids in this subsection. Note that the generation profile $\mathbf{P}_g$ is fixed during the analysis.
First, plugging $\boldsymbol\theta=\mathbf{SP}$ into \eqref{utility} yields
\begin{equation}\label{utility2}
\tilde J_{f,i}(P_i^g,P_{-i}^g)=\psi_i P_i^g+\zeta(P_i^l-P_i^g)+\frac{1}{2}\eta_i^2 (\sum_{j\in \mathcal{N}} s_{ij} P_j)^2,
\end{equation}
for $i\in \mathcal{N}_d$, where $\tilde J_{f,i}:\mathcal{P}^{\mathcal{G}}_{f,i} \times \mathcal{P}^{\mathcal{G}}_{f,-i}\rightarrow \mathbb{R}$, and it is strictly convex on $P_i^g$. Therefore, the first-order optimality condition of \eqref{utility2} yields
$
\psi_i-\zeta+\eta_i^2(\sum_{j\in \mathcal{N}} s_{ij} P_j)s_{ii}=0.
$
By defining $g_i:=s_{ii}P_i$, $\bar{g}_{-i}:=\sum_{j\neq i\in \mathcal{N}}s_{ij}P_j$ and since $s_{ii}\neq 0$ by Lemma \ref{lemma_Smatrix}, we obtain
$g_i=\frac{\zeta-\psi_i}{\eta_i^2 s_{ii}} - \bar{g}_{-i}$
which can be rewritten as
$
P_i=\frac{1}{s_{ii}}(\frac{\zeta-\psi_i}{\eta_i^2 s_{ii}} - \bar{g}_{-i}),\ \ i\in \mathcal{N}_d.
$
For clarity, define 
$\gamma_i:=\frac{\zeta-\psi_i}{\eta_i^2 s_{ii}}$
and $P_i^{\max}:=P_{i,\max}^g-P_i^l$. Then, we further obtain
\begin{equation}\label{bestresponse}
\mathbf{H}\mathbf{P}_d^{*}=\mathbf{q},
\end{equation}
where $\mathbf{H}:=[ \frac{s_{ij}}{s_{ii}} ]_{i,j\in\mathcal{N}_d} ,\ \mathbf{q} := [q_i]_{i\in \mathcal{N}_d} = [ \frac{\gamma_i}{s_{ii}} -\sum_{j\in \mathcal{N}_g} \frac{s_{ij}}{s_{ii}} P_j]_{i\in \mathcal{N}_d} $, and $\mathbf{P}_d^{*}:= [P_i^*]_{i\in \mathcal{N}_d}=[P_i^{g*}-P_i^l]_{i\in \mathcal{N}_d}$.

\textbf{\textit{Remark:}} For a given $\mathbf{P}_g^g$, there exists a one-to-one mapping between $\mathbf{P}_d^g$ and $\boldsymbol\theta_d$ through \eqref{power_dc}. Hence, the strategy pair $(\mathbf{P}_d^g,\boldsymbol\theta_d)$ can be equivalently captured by $\mathbf{P}_d^g$. The same analysis applies for $\mathbf{P}_g^g$ and $(\mathbf{P}_g^g,\boldsymbol\theta_g)$ when given $\mathbf{P}_d^g$. In the following, we refer to $\mathbf{P}_g^g$ and $\mathbf{P}_d^g$ as the strategies of generators and microgrids, respectively, for clarity.

\begin{lemma}\label{H_invertible}
Matrix $\mathbf{H}$ is invertible.
\end{lemma}
\begin{proof}
See Appendix \ref{app_H_invertible}.
\end{proof}

Regarding the microgrids' strategy, we have Theorem \ref{uniqueNE}.
\begin{theorem}[Theorem 1, \cite{chen2017game}]\label{uniqueNE}
For a given $\mathbf{P}_g$ and an appropriate $\mathbf{S}$, the renewable energy generation game $G_f$ admits a unique Nash equilibrium, and the net power injection of player $i$, $i\in\mathcal{N}_d$, to the grid is given by 
\begin{align}\label{netpower}
P_i=\begin{cases}
\begin{array}{ll}
-P_i^l, &\mathrm{if}\ \gamma_i \leq \bar{g}_{-i}-s_{ii}P_i^l,\\
P_i^{\max}, &\mathrm{if}\ \gamma_i \geq \bar{g}_{-i}+s_{ii}P_i^{\max},\\
\frac{1}{s_{ii}}(\gamma_i - \bar{g}_{-i}), &\mathrm{otherwise}.
\end{array}
\end{cases}
\end{align}
\end{theorem}

\subsection{Generator Side Analysis}\label{generator_analysis}
The generators in the grid are leaders and thus have a complete information of microgrids, such as their power loads, incentives of voltage regulations and generation costs. The inclusion of microgrids in the smart grid is to make the system greener and more efficient. For the followers, zero power generation clearly makes the microgrid highly underutilized, while maximum power generation puts strict constraints on the working hours and reliability issues of microgrids. Therefore, in general, generators will choose strategies that lead the microgrids achieving an inner solution at the equilibrium. Equivalently, generators regard that the best response dynamics of microgrids are given by \eqref{bestresponse}.

Therefore, via Lemma \ref{H_invertible}, the best response of microgrids can be written as $\mathbf{P}_d=\mathbf{H}^{-1}\mathbf{q}$,  and the optimization problem OP for the generators can be reformulated rigorously as
\begin{align*}
\mathrm{OP'}: \quad\min\limits_{P_j^g,\theta_j}\quad &\sum_{j\in\mathcal{N}_g}C_j(P_j^g)+\frac{1}{2}\alpha_j\theta_j^2\\
\mathrm{s.t.}\quad &\mathrm{\mathbf{P}}=-\mathbf{B}\boldsymbol{\theta},\\
&\mathbf{P}_d=\mathbf{H}^{-1}\mathbf{q},\\
&0\leq P_j^g \leq P_{j,\max}^g,\ \forall j\in\mathcal{N}_g.
\end{align*}
To simplify $\mathrm{OP'}$, we have the following Lemma.
\begin{lemma}\label{lemma3}
The constraints $\mathbf{P}=-\mathbf{B}\boldsymbol{\theta}$ and $\mathbf{P}_d=\mathbf{H}^{-1}\mathbf{q}$ can be captured by 
\begin{equation}\label{pg2}
\mathbf{T}_1\boldsymbol{\theta}_{g}+\mathbf{T}_2\mathbf{q}-\mathbf{P}_{g}=0,
\end{equation}
where $\mathbf{T}_1:=\mathbf{B}_3\mathbf{B}_1^{-1}\mathbf{B}_2-\mathbf{B}_4\in \mathbb{R}^{N_g\times N_g}$ and $\mathbf{T}_2:=\mathbf{B}_3\mathbf{B}_1^{-1}\mathbf{H}^{-1}\in \mathbb{R}^{N_g\times N_d}$.
\end{lemma}

\begin{proof}
See Appendix \ref{app_lemma3}.
\end{proof}

One challenge in $\mathrm{OP'}$ is the box constraint for generators, i.e., $0\leq P_j^g \leq P_{j,\max}^g,\ \forall j\in\mathcal{N}_g$. Note that in the objective function of $\mathrm{OP'}$, each generator $j$, $j\in\mathcal{N}_g$, has a weighting parameter $\alpha_j$ indicating the incentive of regulating the voltage angle at his bus. When generator $j$ cares more about the generation cost rather than the voltage regulation, then small $\alpha_j$ is chosen. In this case, due to the monotonically increasing quadratic function $C_j(P_j^g)$, the amount of generated power of generators should be relatively small to optimize the objective, and $0\leq P_j^g \leq P_{j,\max}^g$ is naturally satisfied. We focus on this scenario and aim to obtain analytical results for generators in the following. The leader's problem $\mathrm{OP'}$ can be simplified as
\begin{align*}
\mathrm{OP''}: \quad\min\limits_{P_j^g,\theta_j}\quad &\sum_{j\in\mathcal{N}_g}C_j(P_j^g)+\frac{1}{2}\alpha_j\theta_j^2\\
\mathrm{s.t.}\quad &\mathbf{T}_1\boldsymbol{\theta}_{g}+\mathbf{T}_2\mathbf{q}-\mathbf{P}_{g}=0.
\end{align*}
The Lagrangian of $\mathrm{OP''}$ is
$$L(P_j^g,\theta_j,\boldsymbol{\mu})=\sum_{j\in\mathcal{N}_g}C_j(P_j^g)+\frac{1}{2}\alpha_j\theta_j^2+\boldsymbol{\mu}^T(\mathbf{T}_1\boldsymbol{\theta}_{g}+\mathbf{T}_2\mathbf{q}-\mathbf{P}_{g}),$$
where $\boldsymbol{\mu}:=[\mu_1,\mu_2,...,\mu_{N_{g}}]$ is the Lagrange multiplier vector.
The first-order optimality condition $\nabla L=0$ yields
\begin{align}
\frac{\partial L}{\partial P_j^g}&=a_jP_j^g+b_j-\mu_j-\sum_{i\in\mathcal{N}_d}\boldsymbol{\mu}^T {T}_2(i) \frac{s_{ij}}{s_{ii}} =0,\label{kkt1}\\
 \frac{\partial L}{\partial \theta_j}&=\alpha_j\theta_j+\boldsymbol{\mu}^T {T}_1(j)=0,\ \forall j\in\mathcal{N}_g,\label{kkt2}
\end{align}
where $\mathbf{T}_1=[{T}_1(1),...,{T}_1(N_g)]$ and $\mathbf{T}_2=[{T}_2(1),...,{T}_2(N_d)]$.
In addition, the complementarity slackness condition is the same as constraint \eqref{pg2}. Then, putting \eqref{pg2}, \eqref{kkt1} and \eqref{kkt2} in a matrix form yields
\begin{equation}\label{system}
\begin{split}
\mathbf{WX} = \mathbf{b},
\end{split}
\end{equation}
where $\mathbf{X}=[{\mathbf{P}_g^g}^T,\boldsymbol{\mu}^T,{\boldsymbol{\theta}_g}^T]^T;\ \mathbf{0}=[0]_{N_g\times N_g}$; $\mathbf{I}$ is an $N_g$-dimensional identity matrix;
\begin{align*}
&\mathbf{W}=\begin{bmatrix}
\mathbf{A}_1 &\mathbf{T}_3-\mathbf{I}&\mathbf{0}\\
\mathbf{0}&\mathbf{T}_1^T&\mathbf{A}_2\\
\mathbf{T}_4-\mathbf{I}&\mathbf{0}&\mathbf{T}_1
\end{bmatrix};\\
& \mathbf{A}_1 = \mathrm{diag}(a_1,a_2,...,a_{N_g});\ \mathbf{A}_2 = \mathrm{diag}(\alpha_1,\alpha_2,...,\alpha_{N_g});\\
&\mathbf{b}=[-b_1,...,-b_{N_g},0,...,0,T_5(1),...,T_5(N_g)]^T;\\
&T_3(i,j)=-\sum_{p\in \mathcal{N}_d}T_2(j,p) \frac{s_{pi}}{s_{pp}} ;\\ 
&T_4(i,j)=-\sum_{p\in \mathcal{N}_d}T_2(i,p) \frac{s_{p(p+j)}}{s_{pp}} ;\\
&T_5(i)=-\sum_{q\in\mathcal{N}_d}T_2(i,q)\frac{\gamma_q}{s_{qq}},\ \forall i,j\in\{1,2,...,N_g\}.
\end{align*}
Hence, to solve $\mathrm{OP''}$, we need to solve the system of equations \eqref{system}. The invertibility of $\mathbf{W}$ is crucial for the uniqueness of SE solution, and we have the following lemma.

\begin{lemma}\label{invertible_W}
Matrix $\mathbf{W}$ in \eqref{system} is invertible.
\end{lemma}
\begin{proof}
See Appendix \ref{app_invertible_W}.
\end{proof}

Based on Lemma \ref{invertible_W}, we obtain the following corollary.

\begin{corollary}
A unique Stackelberg equilibrium solution exists in the game $G$, since matrix $\mathbf{W}$ is invertible.
\end{corollary}
\begin{proof}
From Theorem \ref{uniqueNE}, we know that each microgrid has a unique solution in the Nash game $G_f$ for a given $\mathbf{P}_g$. In addition, since $\mathbf{W}$ is invertible, generators admit a unique solution $\mathbf{P}_g$ in the leader's problem. Therefore, game $G$ possesses a unique Stackelberg equilibrium.
\end{proof}

This section has characterized the SE solution to game $G$ through analysis on both microgrid and generator sides. In the ensuing section, we will discuss methodologies and design algorithms to find the SE.

\section{Distributed Update Schemes}\label{algorithms}
We know that game $G$ has a unique equilibrium solution from Section \ref{gameanalysis}. In this section, we aim to design update schemes for the generators and microgrids, respectively, to compute the Stackelberg equilibrium strategy.

\subsection{Microgrid Side Update Schemes}\label{microgrid_algorithms}
The update schemes of microgrids mainly follow the previous work \cite{chen2017game}, and we present them briefly as follows for completeness. 

\subsubsection{Iterative Update Algorithm (IUA)}
The IUA is a scheme that all microgrids update their strategies simultaneously. At time step $n$, the update for microgrid $i$, $i\in\mathcal{N}_d$, is
\begin{equation}\label{iterativeupdate}
\begin{split}
P_i^{(n+1)}&=\Psi_i(\gamma_i,\bar g_{-i}^{(n)})=\min\Big( P_i^{\max},\ \max \big[ -P_i^l,\\
&\qquad\frac{1}{s_{ii}}(\gamma_i - \sum\limits_{j\in \mathcal{N}_g}s_{ij}P_j 
-\sum\limits_{j\neq i\in \mathcal{N}_d}s_{ij}P_j^{(n)} )\big]\Big).
\end{split}
\end{equation}

\subsubsection{Random Update Algorithm (RUA)}
The IUA requires that every player updates their decisions in parallel which is impractical in cases without synchronization mechanism. One more practical update scheme is RUA. Specifically, microgrids update their strategies with a predefined probability $0<\tau_i<1,\ i\in\mathcal{N}_d$ at each step. The RUA is given by
\begin{equation*}\label{randomupdate}
P_i^{(n+1)}= \begin{cases}
\begin{array}{ll}
{\Psi_i(\gamma_i,\bar g_{-i}^{(n)})}, &\mathrm{with\ probability}\ \tau_i,\\
P_i^{(n)}, &\mathrm{with\ probability}\ 1-\tau_i,
\end{array}
\end{cases}
\end{equation*}
where $\Psi_i$ is defined in \eqref{iterativeupdate}.

\subsubsection{PMU-Enabled Distributed Algorithm (PDA)}\label{tech_algorithm}
The PDA does not require the synchronization mechanism as that in IUA. Its update fashion is similar to RUA but requires much less information. Notice that $\boldsymbol\theta=\mathbf{SP}$ leads to
$
\sum\limits_{j\in \mathcal{N}_g}s_{ij}P_j+\sum\limits_{j\neq i\in \mathcal{N}_d}s_{ij}P_j=\theta_i-s_{ii}P_i,\ \forall i\in\mathcal{N}_d
$. By 
incorporating it into \eqref{iterativeupdate}, each microgrid can update the decision by knowing the voltage angle at his bus at the current step.  Therefore, the PDA for player $i$, $i\in\mathcal{N}_d$, is 
\begin{equation}\label{pmu_update}
\begin{split}
P_i^{(n+1)}=\min\Big( P_i^{\max},\ &\max \big[ -P_i^l,\\
&\frac{1}{s_{ii}}(\gamma_i -\theta_i^{(n)} +s_{ii}P_i^{(n)} )\big]\Big).
\end{split}
\end{equation}


A critical property of the update algorithm is its convergence. Specifically, we have the following theorem.
\begin{theorem}[Theorem 3, \cite{chen2017game}]\label{thm3}
The PMU-enabled distributed algorithm of microgrid is globally stable and can converge to the unique equilibrium point almost surely if
\begin{align}
\bar{\tau}\cdot \max_{i,j\neq i\in \mathcal{N}_d} \frac{s_{ij}}{s_{ii}}(N_d-1)<\underline{\tau},\label{randomthem}
\end{align}
where $\bar{\tau}$ and $\underline{\tau}$ denote the upper and lower bounds of probability $\tau_i$, respectively, $\forall i\in \mathcal{N}_d$.
\end{theorem}


\subsection{Generator Side Update Schemes}\label{generator_algorithm}
For the generators in the smart grid, their equilibrium strategies are given by $\mathbf{WX} = \mathbf{b}$ in \eqref{system}. To obtain the solution, we adopt Guass-Seidel iterative method \cite{nocedal2006numerical} when the inverse of matrix $\mathbf{W}$ is complex to compute. Specifically, the update scheme is given by
\begin{align}\label{GS}
X_i^{(t+1)} = \frac{1}{W_{ii}}\big[ b_i - \sum_{k=1}^{i-1} W_{ik}X_k^{(t+1)} - \sum_{k=i+1}^{3N_g} W_{ik}X_k^{(t)} \big]
\end{align}
for $i = 1,2,...,3N_g$, where $t$ is the time step index, and $\mathbf{W}=[W_{ik}]_{i,k\in\{1,2,...,3N_g\}}$. Putting \eqref{GS} in a matrix format yields
\begin{align}
\mathbf{X}^{(t+1)} = \mathbf{M}\mathbf{X}^{(t)} + \mathbf{D}^{-1}\mathbf{b},\label{matrixform}
\end{align}
where $\mathbf{M} = \mathbf{D}^{-1}(\mathbf{D}-\mathbf{W})$, and $\mathbf{D}= \begin{bmatrix}
\mathbf{A}_1&\mathbf{0}&\mathbf{0}\\
\mathbf{0}&\mathbf{T}_1^{RU}&\mathbf{0}\\
\mathbf{T}_4-\mathbf{I}&\mathbf{0}&\mathbf{T}_1^{LL}
\end{bmatrix}$. $\mathbf{D}$ includes the lower triangular portion of matrix $\mathbf{W}$. Note that sub-matrices $\mathbf{T}_1^{RU}$ and $\mathbf{T}_1^{LL}$ in $\mathbf{D}$ represent the upper-right and lower-left triangular matrices of $\mathbf{T}_1$, respectively. The invertibility of matrix $\mathbf{D}$ in \eqref{matrixform} is critical. We have the following lemma.

\begin{lemma}\label{lemma2}
Matrix $\mathbf{D}$ is invertible, and thus \eqref{matrixform} is feasible.
\end{lemma}

\begin{proof}
See Appendix \ref{app2}.
\end{proof}

Since \eqref{GS} is a linear difference equation, a necessary and sufficient condition for its convergence to the unique equilibrium point is that the eigenvalues of matrix $\mathbf{M}$ are in the unit circle, that is,
\begin{align}\label{convergence}
\rho (\mathbf{M})<1,
\end{align}
where $\rho(\cdot)$ denotes the spectral radius operator. Note that matrix $\mathbf{M}$ is related to the smart grid topology and the constant parameters in the leader's objective. Therefore, when those grid and objective parameters are appropriate such that matrix $\mathbf{M}$ satisfies \eqref{convergence}, the iterative algorithm \eqref{GS} converges. Under the case that \eqref{convergence} does not hold, then the generators can directly solve $\mathbf{WX} = \mathbf{b}$ through $\mathbf{X}=\mathbf{W}^{-1}\mathbf{b}$, since matrix $\mathbf{W}$ is invertible by Lemma \ref{invertible_W}. 

Besides the public known power system topology, the only required information from the followers during the updates of \eqref{GS} is $\gamma_i$ in $\mathbf{b}$ which is related to the unit renewable generation cost $\psi_i$ and the incentive of voltage regulation $\eta_i$, for $i\in\mathcal{N}_d$. In the following, we present three methods with different \textit{information structures} to enable the leader's updates.

\subsubsection{Knowledge of Private Parameters (KPP)} The first approach is that microgrids inform generators with their parameters $\psi_i$ and $\eta_i$, $\forall i\in\mathcal{N}_d$, in advance. Then, generators can determine their strategy $\mathbf{P}_g^{*}$ by using \eqref{GS} directly without any further interactions with microgrids. The drawback of this method is that generation costs and incentives of voltage regulation are private information of microgrids, and thus may not be available to generators.

\subsubsection{Knowledge of Generation Decisions (KGD)}
To address the privacy issue in KPP, an alternative approach to obtain parameters $\gamma_i$, $\forall i\in\mathcal{N}_d$, is that generators first announce their generations $\mathbf{P}_g$ to microgrids, and then microgrids send the best response renewable energy profile $\mathbf{P}_d$ back to generators. The information about power generations is much less private  than the specific parameters, such as $\psi_i$ and $\eta_i$, $\forall i\in\mathcal{N}_d$. In addition, remind that $\mathbf{q} = \big[ \frac{\gamma_i}{s_{ii}} -\sum_{j\in \mathcal{N}_g}  \frac{s_{ij}}{s_{ii}} P_j\big]_{i\in \mathcal{N}_d}$, then $\gamma_i$ can be obtained from $\mathbf{q} = \mathbf{H}\mathbf{P}_d$.

\subsubsection{Knowledge of Bus Voltage Angles (KBA)} Though KGD preserves high confidentiality for microgrids, leaders still need the feedback information of amount of renewable generations from the followers, and thus requires a \textit{two-way} smart grid communication infrastructure. Remind that the communication channels from generators to microgrids can be used to transmit information including the power generation and market power price. To reduce the communication cost, we propose a fully distributed update scheme for the generators, and therefore simplifies the two-way communication model to a single-way one.

The main focus is to design a distributed mechanism to obtain parameters $\gamma_i$, for $i\in\mathcal{N}_d$. Note that we have $\mathbf{T}_1\boldsymbol{\theta}_g+\mathbf{T}_2\mathbf{q}-\mathbf{P}_g=0$ which can be rewritten as 
\begin{equation}\label{estimation}
\mathbf{T}_2 {\boldsymbol\Upsilon} =  \mathbf{P}_g-\mathbf{T}_1\boldsymbol{\theta}_g + \mathbf{T}_2 \boldsymbol\Lambda,
\end{equation}
where $\boldsymbol\Upsilon:=\big[ \frac{\gamma_i}{s_{ii}} \big]_{i\in \mathcal{N}_d}$, and $\boldsymbol\Lambda:= \big[ \sum_{j\in \mathcal{N}_g}  \frac{s_{ij}}{s_{ii}} P_j \big]_{i\in \mathcal{N}_d}$.
In the Stackelberg game, when generators announce their strategy $\mathbf{P}_g$ to microgrids, they can obtain a voltage angle profile $\boldsymbol{\theta}_g$ by using PMUs after the best response of microgrids.  Therefore, the right hand side of \eqref{estimation} is a complete information to the generators.  However, to obtain vector $\boldsymbol\Upsilon$,
another challenge is that the number of generators is generally less than the number of microgrids, i.e., $N_g<N_d$. Thus, $\mathbf{T}_2$ is an underdetermined matrix, and solving \eqref{estimation} gives an infinite number of solutions if there are any.

To address this problem, one possible way is to seek an estimation of $\hat{\boldsymbol\Upsilon}$ based on \eqref{estimation}. However, to remove the inconsistency between the measured voltage angle $\boldsymbol{\theta}_g$ and the anticipated one due to the parameter estimation error is an obstacle and hence makes this method infeasible. Realizing that separate parameters $\gamma_i$, $\forall i\in\mathcal{N}_d$, are challenging to obtain, we turn to investigate the form of $T_5(j),$ for ${j\in\mathcal{N}_g}$. Remind that for $i\in\mathcal{N}_g$,
$
T_5(i)=-\sum_{q\in\mathcal{N}_d}T_2(i,q)\frac{\gamma_q}{s_{qq}}.
$
Then, for convenience, we define 
\begin{equation}\label{unknownT5}
\tilde T_{5}(i):=\sum_{q\in\mathcal{N}_d}T_2(i,q)\frac{\gamma_q}{s_{qq}},\ i\in\mathcal{N}_g,
\end{equation}
 which is a weighted aggregation of $\gamma_q,\ \forall q\in\mathcal{N}_d$. Denote $\tilde{\mathbf{T}}_5=[\tilde T_{5}(i)]_{i\in\mathcal{N}_g}$, and notice that vector $\mathbf{T}_2 {\boldsymbol\Upsilon}$ on the left hand side of \eqref{estimation} is equivalent to $\tilde{\mathbf{T}}_5$, i.e., $\mathbf{T}_2 {\boldsymbol\Upsilon} = \tilde{\mathbf{T}}_5$. Since the right hand side of \eqref{estimation} is known to the generators, then, the unknown information $\tilde{\mathbf{T}}_5$ can be obtained subsequently. 

Based on the chosen strategy $\mathbf{P}_g$ and the corresponding best response ${\boldsymbol\theta}_g$ measured by PMUs, generators can calculate $\mathbf{P}_g-\mathbf{T}_1\boldsymbol{\theta}_g + \mathbf{T}_2 \boldsymbol\Lambda$ and assign it to $\tilde{\mathbf{T}}_5$. Then, vector $\mathbf{b}$ in \eqref{matrixform} is determined, and thus generators can update their strategies through algorithm \eqref{GS} in a fully distributed fashion. The communications from microgrids to generators are not required which make KBA more advantageous than KPP and KGD. The single-way information flow not only reduces the communication construction costs, but also improves the quality of service due to the enhancement of signal-to-noise-ratio of channels. Particularly, in cases where the microgrids use PDA, and also they  know the market power price from the Internet which is an adopted mechanism currently, then the smart grid communication infrastructures between two layers and among microgrids are not needed.

\subsection{Combined Distributed Update Algorithm}
We have proposed three update schemes for microgrids and generators, respectively. Different combinations of these methods have various features. We present some key characteristics including communication cost, privacy level and algorithm update efficiency of each combination in Table \ref{t1}. Specifically, the communication cost for each scheme is measured through the number of communication links between the players in the grid. By using KBA, the generators do not need to know microgrid parameters $\psi_i$ and $\eta_i$ (KPP), $\forall i\in\mathcal{N}_d$, and the generation decisions of microgrids (KGD), and hence reduces the number of communication links from the microgrids to the generators. From the microgrids' side, PDA does not require communication links from the generators to the microgrids and the ones between microgrids themselves as IUA and RUA do. Therefore, the implementations of strategies combing IUA/RUA with KPP/KGD require the largest communications investment in terms of the total number of links. By using IUA\&{KPP} as the benchmark, we can measure the communication costs of other strategies. For example, using KBA\&{PDA} strategy saves the communication infrastructures between generator and microgrid layers as well as
those within microgrids. Therefore, the communication cost is ultra low. In contrast, PDA\&KPP yields a low communication cost by reducing the number of required communications between microgrids.

The privacy level of each strategy is determined by the type of required information for updates of each agent. The local cost parameters of $\psi_i$ and $\eta_i$, $\forall i\in\mathcal{N}_d$, are direct private and sensitive information of a microgrid. In addition, individual decisions of the other players are viewed as indirect private information. However, PMU measurements on the buses are viewed as the common information shared between the generator and the microgrids. Therefore, the scheme, PDA\&{KBA}, adopts PMUs that measure the bus angles to enable the decision updates of players preserves the highest level of privacy. Note that at the microgrids’ side, when a microgrid uses IUA or RUA strategy, it needs to know the decisions of the other players in the grid as well. Using the
combined strategy PDA\&{KBA} as the benchmark, we can determine the privacy level of other strategies including IUA/RUA and KPP/KGD. For example, the
combined update scheme with PDA\&{KGD} has to disclose generation decisions, and hence its privacy level can be deemed as high instead of ultra high. The worst cases in terms of privacy are IUA/RUA and KPP since they disclose private parameters $\psi_i$ and $\eta_i$, $\forall i\in\mathcal{N}_d$, and the generation decisions at the same time. We can further rank the privacy level of other strategies in a similar fashion.  By comparison, the scheme combing PDA and KBA is preferred, since its communication cost is ultra low, and it preserves high privacy for microgrids. The update efficiency is measured via the time required by the system to reach an SE point which is mainly affected by the update fashion of microgrids. Generally, the update efficiency of IUA outperforms those of RUA and PDA. However, when update probability $\tau_i$ of microgrid $i$ is large, then the efficiency of IUA, RUA and PDA are without significant difference. 

\begin{table}[t]
\centering
\renewcommand\arraystretch{1.3}
\caption{Characteristics of the Combined Update Schemes\label{t1}}
\begin{tabular}{|c|c|c|c|} \hline
\diagbox[width=8em]
  {Strategy}{Metric}&\thead{Communication\\ Cost}& \thead{Privacy\\ Level}&\thead{Update\\Efficiency} \\ \hline
{IUA \& KPP} & High & Low & Ultra High \\ \hline
{IUA \& KGD} & High & Medium & High \\ \hline
{IUA \& KBA} & Medium & Medium & High \\ \hline
{RUA \& KPP} & High & Low & Medium \\ \hline
{RUA \& KGD} & High & Medium & Medium \\ \hline
{RUA \& KBA} & Medium & Medium & Medium \\ \hline
{PDA \& KPP} & Low & Medium & Medium \\ \hline
{PDA \& KGD} & Low & High & Medium \\ \hline
{PDA \& KBA} & Ultra Low & Ultra High & Medium \\ \hline
\end{tabular}
\end{table}

By combining the PDA for microgrids and the Guass-Seidel iterative method incorporating with KBA for generators, we arrive at a fully distributed update scheme to compute the Stackelberg equilibrium solution in the two-level smart grid. For clarity, the scheme is summarized in Algorithm 1.

\textit{\textbf{Remark}}: The convergence of Algorithm \ref{algorithm3} is guaranteed when Theorem \ref{thm3} for microgrids and condition \eqref{convergence} for generators are  both satisfied.

\begin{algorithm}[t]
\caption{Distributed Scheme to Search for SE Strategy}\label{algorithm3}
\begin{algorithmic}[1]
\State Initialize $P_i^l,\ P_i^{\max},\ \tau_i,\ \forall i\in\mathcal{N}_d$, $P_{j,\max}^g,\ \forall j\in\mathcal{N}_g$, error tolerance $\epsilon_1,\ \epsilon_2$
\State For $n=1$ and $t=1$, arbitrarily choose feasible $P_i^{(n)}$ for $i \in \mathcal{N}_d$, and $P_i^{(t)}$ for $i \in \mathcal{N}_g$
\State \textbf{Repeat} for $n$ if needed \label{innerloop}
\State \textbf{for} $i=1,2,...,N_d$
\State Microgrid $i$ measures its bus voltage angle $\theta_i^{(n)}$ by PMU
\State Random update $P_i^{(n+1)}$ through distributed scheme \eqref{pmu_update}
\State \textbf{end for} 
\State \textbf{if} $\lVert \mathbf{P}_d^{g(n+1)} - \mathbf{P}_d^{g(n)} \rVert_\infty > \epsilon_1$
\State \quad $n=n+1$, and go to step \ref{innerloop}
\State \textbf{else}
\State \quad Go to step \ref{outterloop}
\State \textbf{end}
\State If $t=1$, generators measure $\boldsymbol{\theta}_g^{(1)}$ by PMUs, then calculate $\mathbf{P}_g^{(1)}-\mathbf{T}_1\boldsymbol{\theta}_g^{(1)} + \mathbf{T}_2 \boldsymbol\Lambda^{(1)}$ and assign it to $\tilde{\mathbf{T}}_5$. Obtain unkown parameters via \eqref{unknownT5}. Otherwise, go to step \ref{out}\label{outterloop}
\State \textbf{Repeat} for $t$ if needed 
\State \textbf{for} $j=1,2,...,3N_g$
\State Update $X_j^{(t+1)}$ through \eqref{GS}
\State \textbf{end for} 
\State \textbf{if} $\lVert \mathbf{P}_g^{g(t+1)} - \mathbf{P}_g^{g(t)} \rVert_\infty > \epsilon_2$
\State \quad $t = t+1$, and go to step \ref{outterloop}
\State \textbf{else}
\State \quad $\mathbf{P}_g^{g*} = \mathbf{P}_g^{g(t+1)}$, and go to step \ref{innerloop}
\State \textbf{end}
\State $\mathbf{P}_d^{g*}=\mathbf{P}_d^{g(n+1)}$ \label{out}
\State \textbf{return} $\mathbf{P}_d^{g*}$ and $\mathbf{P}_g^{g*}$
\end{algorithmic}
\end{algorithm}

\begin{figure*}[t]
  \centering
  
{%
    \includegraphics[width=0.85\textwidth]{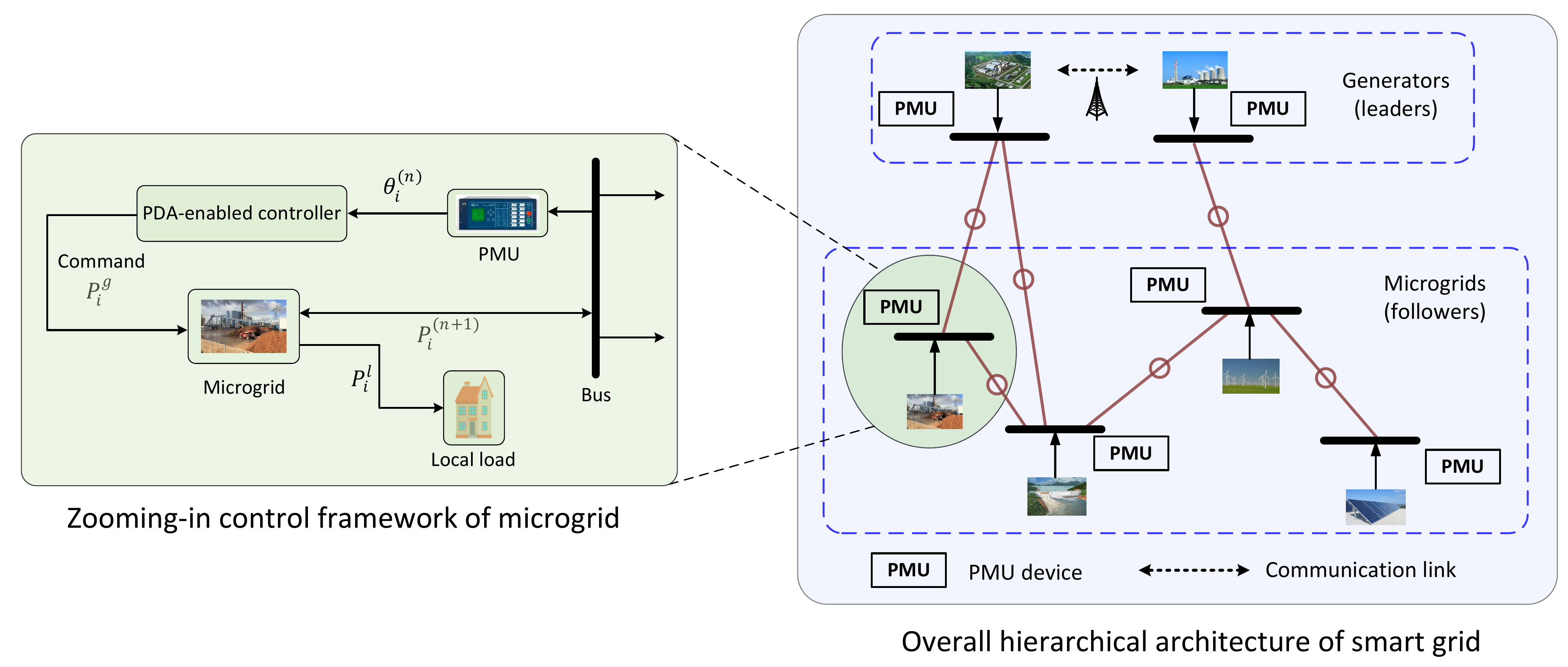}%
    \label{gene_down}%
  }%
  \caption{Illustration of the overall implementation architecture of Algorithm \ref{algorithm3}. Generators make decisions collaboratively via internal communications. PMUs are installed at buses which enable the updates.  In addition, the local control framework of microgrids based on the measured bus voltage angle is designed.}
  \label{framework_distributed}
\end{figure*}

\section{Implementation Architectures of Algorithms in Smart Grid}\label{architecture}
In this section, we develop implementation architectures for the proposed algorithms in smart grid.

\subsection{Communication-Based Implementation Architecture}
Different schemes developed in Section \ref{algorithms} require various implementation frameworks. Specifically, for microgrid's update, IUA and RUA need the support of communication channels both from generators to microgrids and among microgrids themselves. For generator's update, communications from microgrids to generators are necessary for KPP and KGD enabled algorithms. Communication networks in the smart grid can be realized by a number of technologies, such as Internet protocol, power line communication (PLC), WiMAX, ZigBee, optical-fiber and cellular networks \cite{hossain2012smart}. Though PLC is ubiquitous in the power systems, and it has low installation cost, there exists high noise over power lines which causes distortion of signals. The application of wireless technologies can lead to packet losses due to signal interferences and hence impact the efficiency of update algorithms. If the transmitted data between two layers in the framework is modified due to cyber attacks, then the equilibrium outcome of the generation game $G$ would be different. Since the distance between players can be large, the Internet-based networks are more suitable for our framework in terms of signal attenuation and capacity.  The communication-based implementation architecture is shown in Fig. \ref{layer1} in which both generators and microgrids rely on the communication infrastructures for decision updates.

\subsection{PMU-Based Implementation Architecture}
When buses in the smart grid are equipped with PMUs, the convenient PMU-enabled and KBA-enabled distributed algorithms can be deployed by microgrids and generators for updates, respectively. The combined algorithm reduces communication infrastructures among microgrids and between two layers as shown in Fig. \ref{layer1}. In Algorithm \ref{algorithm3}, generators make use of the measured voltage angles by PMUs only once to determine the unknown parameters. Thus, generators' actions are more related to the offline iterative scheme \eqref{GS} which does not include the real-time measurements of voltage angle at buses. Different from the leader, the decision-making of each follower depends on its bus voltage angle at every update step. Therefore, the local control framework of microgrid needs to incorporate the features of smart grid operation.

The overall implementation architecture of Algorithm \ref{algorithm3} is depicted in Fig. \ref{framework_distributed}. Specifically, PMUs are installed at buses to measure the voltage angle. Compared with the communication-based architecture in Fig. \ref{layer1}, the hierarchical framework in Fig. \ref{framework_distributed} does not include communications both between the two layers and within the lower microgrids network. 
The detailed control framework to implement the PMU-enabled algorithm of microgrids is also shown in Fig. \ref{framework_distributed}. Specifically, the PMU measures the voltage angle $\theta_i^{(n)}$ at bus $i$ at step $n$, and sends it to the PDA-enabled local controller. Then, the controller generates a signal that informs the microgrid to dispatch an appropriate amount of renewable energy to the external grid and local load, respectively.
 The negative $P_i^{(n+1)}$ indicates that microgrid $i$ buys power from the external grid.


\section{Case Studies}\label{cases}
We validate our proposed algorithms via case studies based on a 6-bus system shown in Fig. \ref{6bus}. Specifically, buses 1, 2 and 3 are connected to microgrids that generate wind, solar, geothermal renewable energies, respectively, and form a follower network. Buses 4 and 6 are connected to generators which are leaders in the power generation game. Denote the generators at buses 4 and 6 as generator 1 and generator 2, respectively. In addition, bus 5 is selected as the slack bus. Details of the power system model can be found in \cite{wood2012power}.

 Without loss of generality, we set $\eta_i=\eta=10^3$, $P_{i,\max}^g=100\mathrm{MW},\ \forall i \in\mathcal{N}_d$, and $P_{j,\max}^g=800\mathrm{MW},\ \forall j \in\mathcal{N}_g$. In addition, some parameters of generators are as follows: $\alpha_4=50\times 10^4$, $\alpha_6=30\times 10^4$, $a_4=0.05\$/\rm{MW}^2$, $a_6=0.08\$/\rm{MW}^2$ $b_4=6\$/\rm{MW}$, $b_6=8\$/\rm{MW}$, $c_4=130\$$ and $c_6=120\$$. The selection of parameters related to the maximum generations and the generation costs can refer to Chapter 3 of \cite{wood2012power}.  The error tolerance constants are equal to $\epsilon_1 = \epsilon_2=10^{-3}$. The market electricity price and the unit renewable generation cost of microgrids are equal to $\zeta=140\$/\rm{MWh}$, $\psi_1 = 110\$/\rm{MWh}$, $\psi_2 = 150\$/\rm{MWh}$ and $\psi_3 = 80\$/\rm{MWh}$, respectively\cite{taylor2015renewable}. The local loads at microgrids are given by $P_1^l=220\rm{MW}$, $P_2^l=350\rm{MW}$ and $P_3^l=170\rm{MW}$, respectively. The  subscript number of above parameters corresponds to the bus index in the smart grid.

In the case studies, the update probabilities of three microgrids are chosen as follows: $\tau_1=0.7$, $\tau_2=0.7$ and $\tau_3=0.75$. Thus, $\bar \tau=0.75$ and $\underline{\tau}=0.7$. Note that these probabilities correspond to the frequencies of microgrids' decision-makings. Specifically, larger value of update probability indicates more frequent decision updates. The time scale of microgrid's each update is critical. Specifically, the iterative updates of microgrids are faster than their load dynamics, and hence the loads can be seen as fixed in the game. Before finding the equilibrium solution via the distributed algorithm, we first verify the convergence of the update scheme by checking the sufficient conditions in Theorem \ref{thm3} and \eqref{convergence}. The components in matrix $\mathbf{S}$ that correspond to microgrids constitute a submatrix $\mathbf{S}'$, and based on $\mathbf{S}'$ we obtain $\max_{i,j\neq i\in \mathcal{N}_d} \frac{s_{ij}}{s_{ii}}=0.4246$.
 Then we can calculate $\bar{\tau}\cdot \max_{i,j\neq i\in \mathcal{N}_d} \frac{s_{ij}}{s_{ii}}(N_d-1)=0.637<\underline{\tau}=0.7$. Therefore, the sufficient condition in Theorem \ref{thm3} is satisfied. In addition, the constructed matrix $\mathbf{W}$ is invertible which verifies Lemma \ref{invertible_W}, and we have $\rho (\mathbf{M})= 0.885$ which also meets the condition \eqref{convergence}. Hence, the distributed Algorithm \ref{algorithm3} can converge to the unique equilibrium point in this case study.

  \begin{figure}[t]
\centering
\includegraphics[width=0.75\columnwidth]{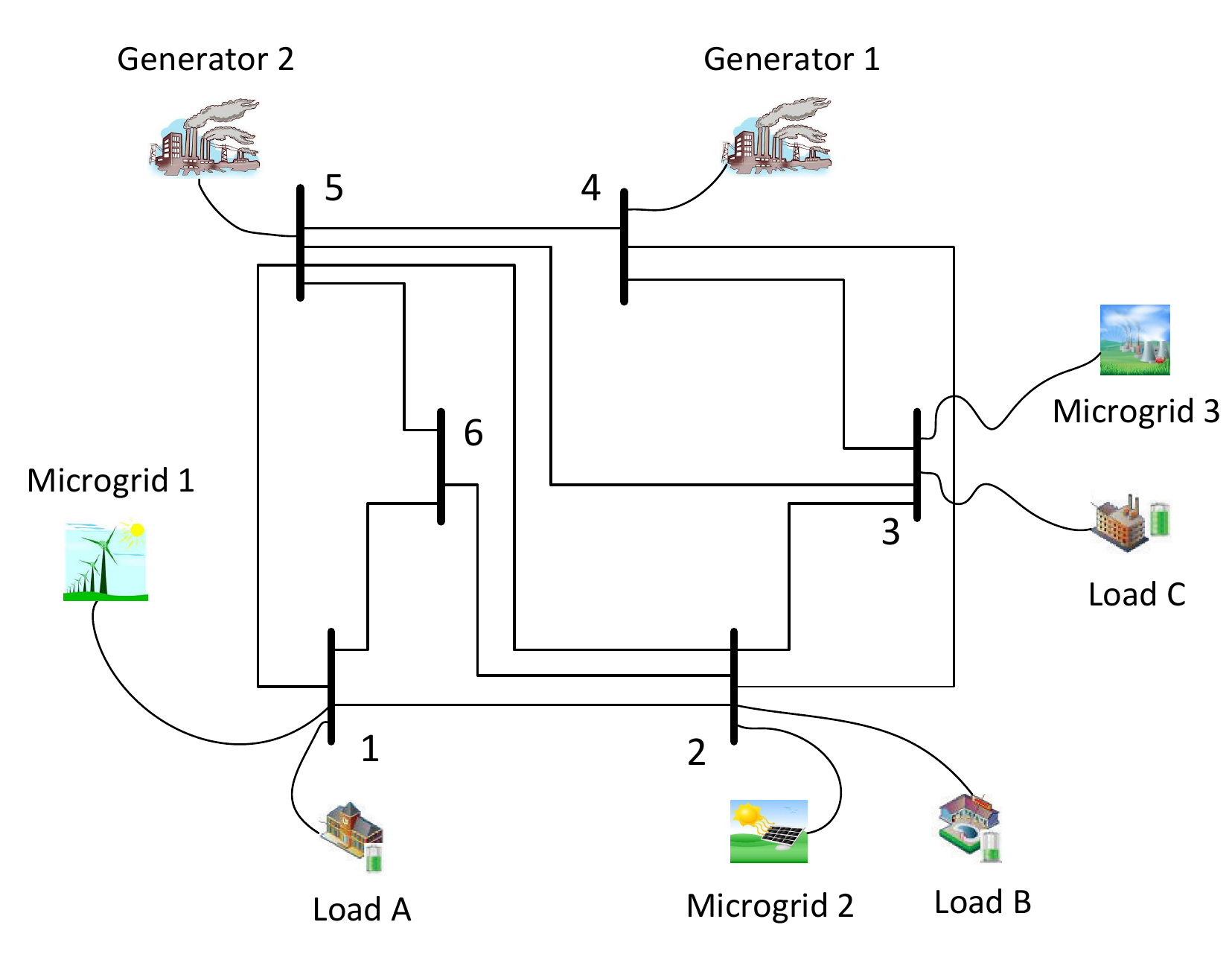}
\caption{6-bus power system model. Buses 1, 2 and 3 are connected to microgrids that generate renewable energies, and they constitute the lower layer. Generators connected to buses 4 and 6 are leaders, and they form the upper layer in the smart grid.}\label{6bus}
\end{figure}

  \begin{figure}[t]
\centering
\includegraphics[width=0.65\columnwidth]{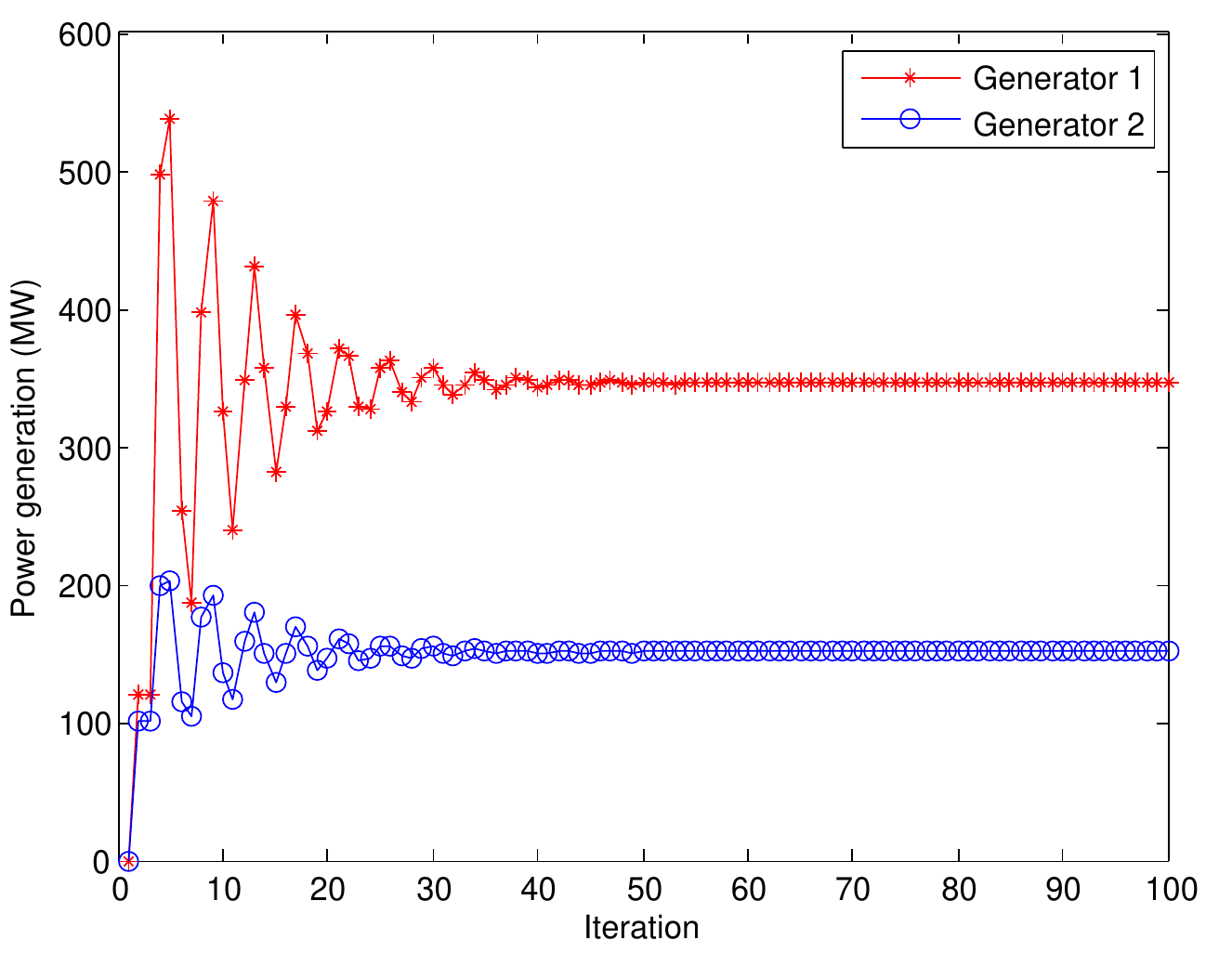}
\caption{The strategy of generators during updates.}\label{gen}
\end{figure}

\begin{figure}[t]
  \centering
  \subfigure[]{
    \includegraphics[width=0.65\columnwidth]{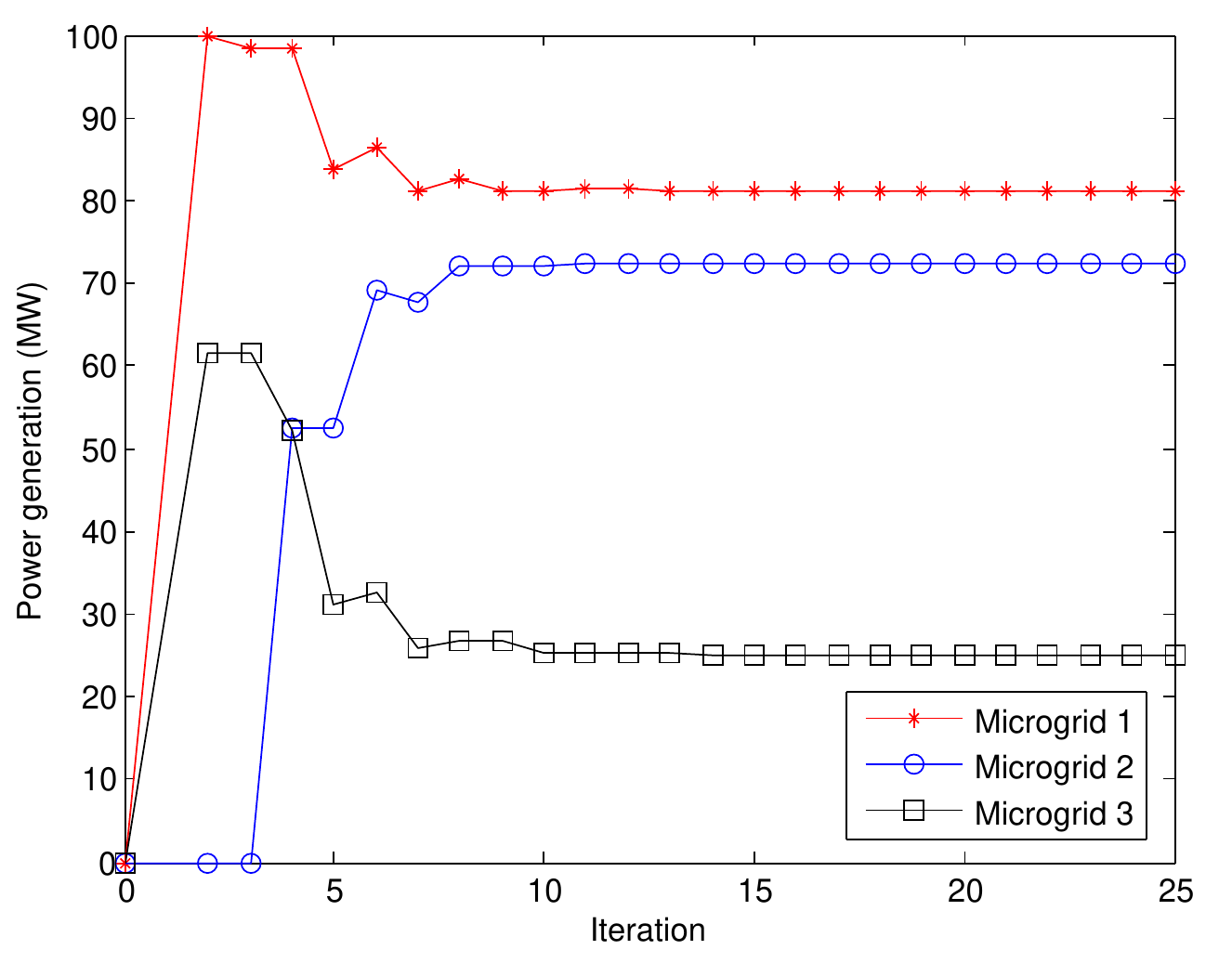}}
	 \subfigure[]{
    \includegraphics[width=0.65\columnwidth]{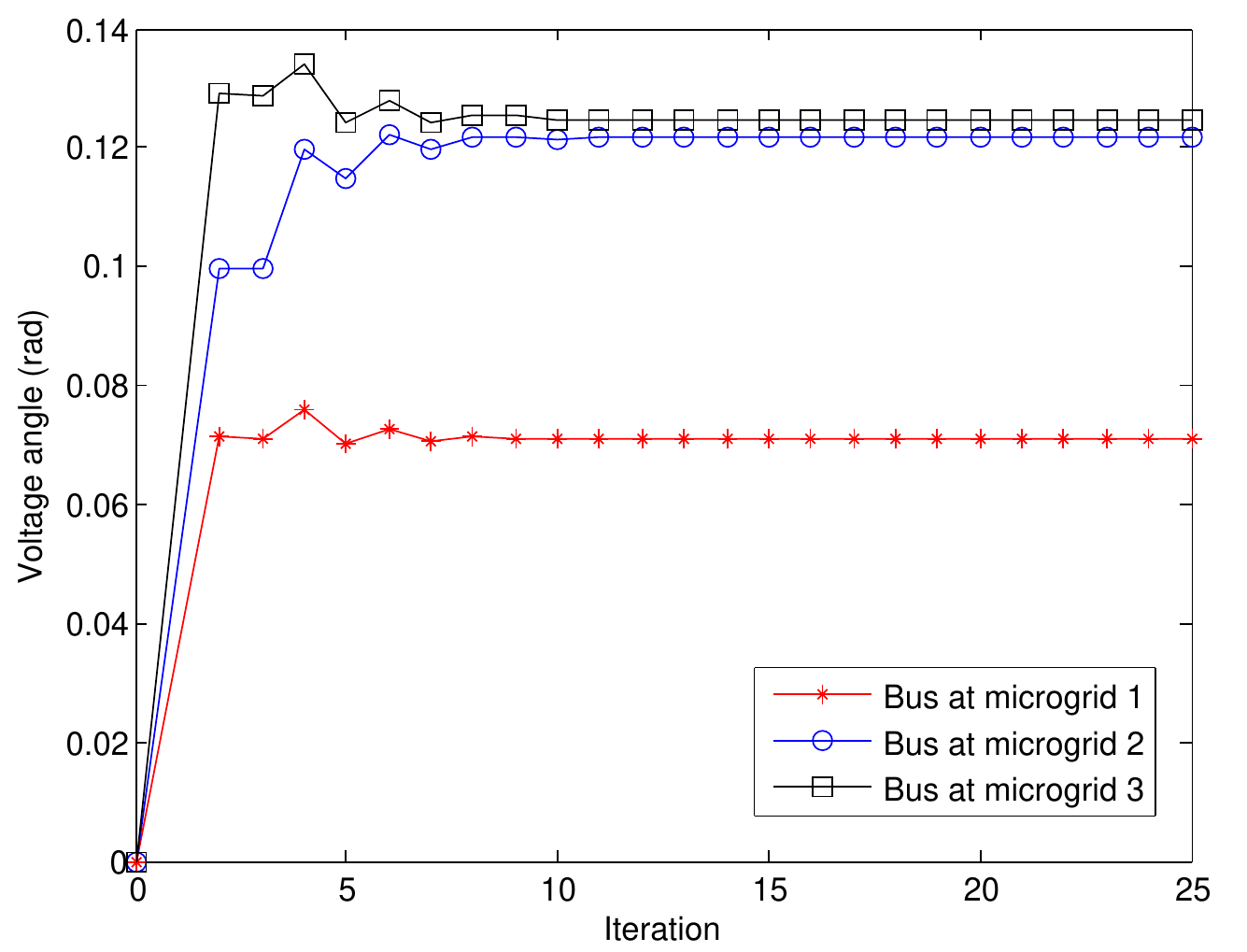}}
  \caption[]{(a) and (b) show the results of the renewable generations and bus voltage angles of microgrids by using the PMU-enabled distributed algorithm, respectively.}
  \label{microgrid}
\end{figure}

\begin{figure}[t]
  \centering
  \subfigure[]{
    \includegraphics[width=0.65\columnwidth]{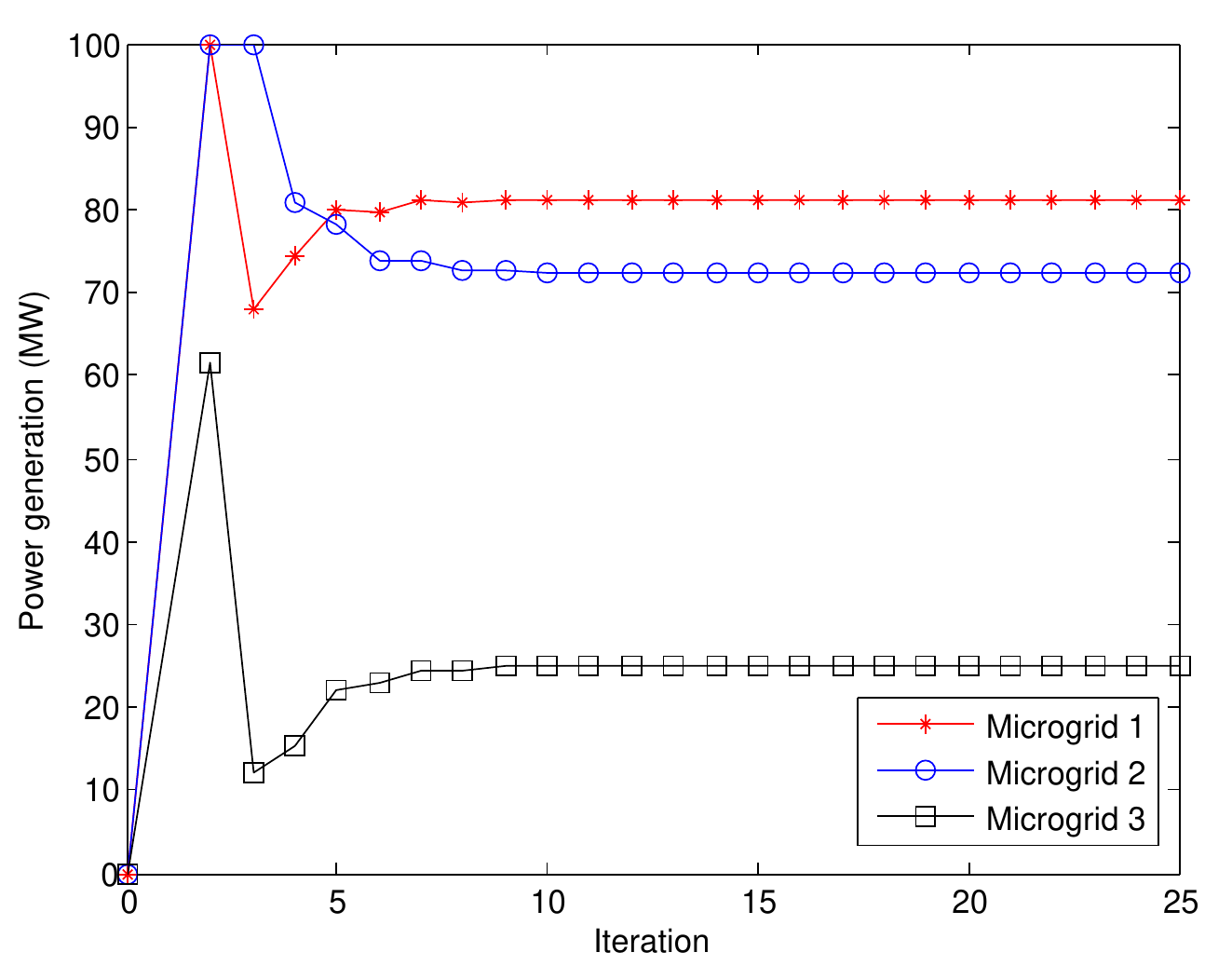}}
	 \subfigure[]{
    \includegraphics[width=0.65\columnwidth]{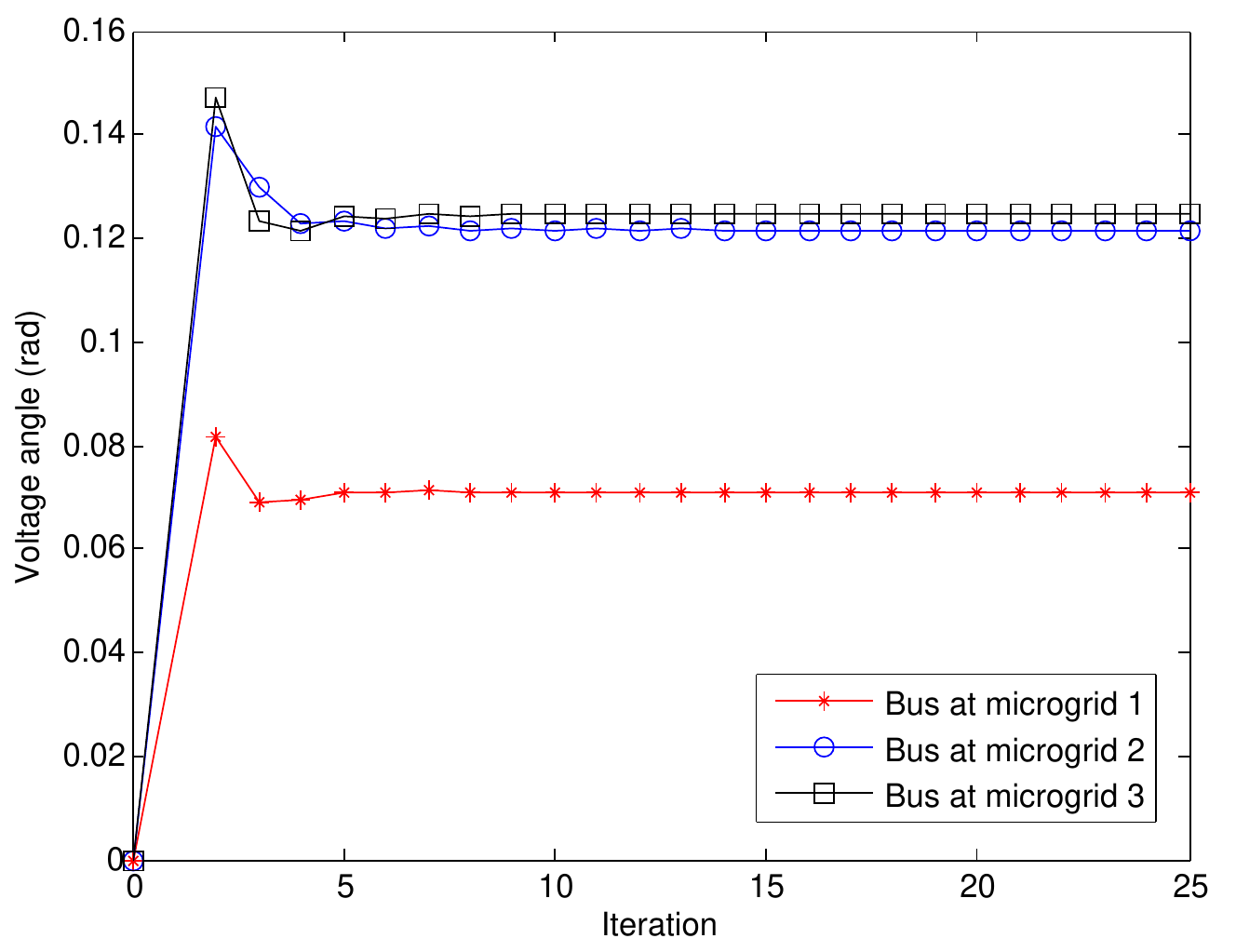}}
  \caption[]{(a) and (b) show the results of the renewable generations and bus voltage angles of microgrids by using the iterative update algorithm, respectively.}
  \label{microgrid_iterative}
\end{figure}

Next, we aim to obtain the Stackelberg equilibrium solution of players in the smart grid by using Algorithm \ref{algorithm3}. The generators first choose an arbitrary feasible decision profile $\mathbf{P}_g^{g(1)}$, and microgrids will respond to it in a fully distributed manner. After the followers reaching an equilibrium, generators measure the voltage angles at their buses and obtain vector $\tilde{\mathbf{T}}_5=[610.5,186.3]^T$ which completes the unknown information extraction through KBA for the leader network. Then, the generators can determine their best strategies via offline updates \eqref{GS}, and the results are shown in Fig. \ref{gen}. Note that the time scale of generators' offline decision updates is negligible comparing with that of microgrids' updates, and the generators implement the strategy at equilibrium in Fig. \ref{gen}. We can see that the generators' equilibrium strategy is given by $\mathbf{P}_g^{g*} = [346.3,151.2]^T$MW. 

The change of power generation of leaders will influence the decision-making of followers. After generators making the decision $\mathbf{P}_g^{g*}$, the best response strategies of microgrids including the renewable generations and voltage angles are shown in Fig. \ref{microgrid}. Note that the voltage angles are small, indicating that the DC approximation performs well \cite{zhu2012game,glover2011power,stott2009dc}. The microgrids can reach the equilibrium point in 14 time steps by using the PMU-enabled distributed algorithm, and the equilibrium generation is $\mathbf{P}_d^{g*} = [81.2,72.3,25.1]^T$MW. Note that the time scale of each time step in the algorithm can be 20 minutes when the microgrids are active in updating their strategies in this generation game. Then, the PMU-enabled distributed algorithm reaches an equilibrium within 5 hours in this case study. For comparison, the decision updates of microgrids enabled by the IUA are shown in Fig. \ref{microgrid_iterative}. The IUA converges to the equilibrium after 9 steps which is faster than the distributed PDA. The reason is that all microgrids update synchronously by using the centralized IUA. In addition, the equilibrium strategies yielded by PDA and IUA in Figs. \ref{microgrid} and \ref{microgrid_iterative} are identical. Therefore, the proposed distributed algorithms for generators and microgrids are effective in finding the Stackelberg equilibrium solution in the two-layer interdependent generation game.

\section{Conclusion}\label{conclusion}
We have established a Stackelberg game framework to capture the  power generation decisions of generators and microgrids which are seen as leaders and followers, respectively, in a smart grid. We have proposed fully distributed schemes for both generators and microgrids to search for the equilibrium strategy. The main required knowledge for decision updates is only the voltage angles at buses which can be obtained by PMU. Case studies have corroborated our theoretical findings and designed algorithms. The future research would be considering the cyber security issues in the smart grids and implementing the designed algorithm experimentally.

\appendices

\section{Proof of Lemma \ref{H_invertible}}\label{app_H_invertible}
\begin{proof}
When the equilibrium solution $\mathbf{P}_d^{*}$ is an inner point, to show its uniqueness, one way is to show that matrix $\mathbf{H}$ is invertible. Since $\mathbf{S}$ is of full rank, and base on the Sylvester's criterion, the upper left $N_d$-dimensional square matrix $\mathbf{S}_1$ in $\mathbf{S}$ is also invertible. Note that determinant $|\mathbf{S}_1|\neq 0$ and it satisfies $|\mathbf{S}_1|=|\mathbf{H}|\cdot \prod_{i\in\mathcal{N}_d} s_{ii}$. Because $s_{ii}>0,\ \forall i\in\mathcal{N}_d$, $|\mathbf{H}|\neq 0$ and thus $\mathbf{H}$ is invertible.
\end{proof}

\section{Proof of Lemma \ref{lemma3}}\label{app_lemma3}
\begin{proof}
First, we rewrite the constraint $\mathrm{\mathbf{P}}=-\mathrm{\mathbf{B}}\boldsymbol{\theta}$ as
\begin{align*}
\begin{bmatrix}
\mathbf{P}_d\\[6pt]
\mathbf{P}_g
\end{bmatrix}=
-\begin{bmatrix}
\mathbf{B}_1& \mathbf{B}_2\\[6pt]
\mathbf{B}_3& \mathbf{B}_4
\end{bmatrix}
\begin{bmatrix}
\boldsymbol{\theta}_d\\[6pt]
\boldsymbol{\theta}_g
\end{bmatrix},
\end{align*}
where $\mathrm{\mathbf{B}}:=\begin{bmatrix}
\mathbf{B}_1& \mathbf{B}_2\\
\mathbf{B}_3& \mathbf{B}_4
\end{bmatrix}$, $\mathbf{B}_1\in \mathbb{R}^{N_d\times N_d}$, $\mathbf{B}_2\in \mathbb{R}^{N_d\times N_g}$, $\mathbf{B}_3\in \mathbb{R}^{N_g\times N_d}$, $\mathbf{B}_4\in \mathbb{R}^{N_g\times N_g}$, and $\boldsymbol{\theta}:=[
\boldsymbol{\theta}_d^T,
\boldsymbol{\theta}_g^T]^T
$.
Then, based on the constraints $\mathbf{P}=-\mathbf{B}\boldsymbol{\theta}$ and $\mathbf{P}_d=\mathbf{H}^{-1}\mathbf{q}$, we have 
\begin{align}
\mathbf{P}_g=-\mathbf{B}_3\boldsymbol{\theta}_d-\mathbf{B}_4\boldsymbol{\theta}_g,\label{couple1}\\
\mathbf{H}^{-1}\mathbf{q}=-\mathbf{B}_1\boldsymbol{\theta}_d-\mathbf{B}_2\boldsymbol{\theta}_g.\label{couple2}
\end{align}
From \eqref{couple2}, $\boldsymbol{\theta}_d$ can be expressed as
\begin{equation}\label{thetad}
\boldsymbol{\theta}_d=-\mathbf{B}_1^{-1}(\mathbf{H}^{-1}\mathbf{q}+\mathbf{B}_2\boldsymbol{\theta}_g),
\end{equation}
since $\mathbf{B}_1$ is invertible. Plugging \eqref{thetad} into \eqref{couple1} yields
$\mathbf{P}_g=\mathbf{B}_3\mathbf{B}_1^{-1}(\mathbf{H}^{-1}\mathbf{q}+\mathbf{B}_2\boldsymbol{\theta}_g)-\mathbf{B}_4\boldsymbol{\theta}_g,$
which is equivalent to
$
(\mathbf{B}_3\mathbf{B}_1^{-1}\mathbf{B}_2-\mathbf{B}_4)\boldsymbol{\theta}_g+\mathbf{B}_3\mathbf{B}_1^{-1}\mathbf{H}^{-1}\mathbf{q}=\mathbf{P}_g
$.
\end{proof}

\section{Proof of Lemma \ref{invertible_W}}\label{app_invertible_W}
\begin{proof}
To prove that matrix $\mathbf{W}$ is invertible, one way is to show that all columns in $\mathbf{W}$  are linearly independent and hence $\mathbf{W}$ is full-rank. Then, we aim to show that the columns in $\mathtt{C}_1:=[\mathbf{A}_1,\mathbf{0},\mathbf{T}_4-\mathbf{I}]^T$, $\mathtt{C}_2:=[\mathbf{T}_3-\mathbf{I},\mathbf{T}_1^T,\mathbf{0}]^T$ and $\mathtt{C}_3:=[\mathbf{0},\mathbf{A}_2,\mathbf{T}]^T$ are linearly independent within each matrix and across all others. We can check that $[\mathtt{C}_1,\mathtt{C}_3]$ is of full column rank, since $\mathbf{A}_1$ and $\mathbf{A}_2$ are diagonal matrices with positive entries. Similarly, the next step is to show that $[\mathtt{C}_1,\mathtt{C}_2]$ and $[\mathtt{C}_2,\mathtt{C}_3]$ are also of full column rank. By close investigation, one sufficient condition that yields full column rank of $[\mathtt{C}_1,\mathtt{C}_2]$ and $[\mathtt{C}_2,\mathtt{C}_3]$ is that matrix $\mathbf{T}_1$ is of full rank. Then, our objective is to show that  $\mathbf{T}_1=\mathbf{B}_3\mathbf{B}_1^{-1}\mathbf{B}_2-\mathbf{B}_4$ is invertible. This immediately follows from the facts that $-\mathbf{T}_1$ is the Schur complement of $\mathbf{B}_1$, and $-\mathbf{T}_1$ is invertible if and only if matrix $\mathbf{B}$ is nonsingular \cite{bernstein2009matrix}. Therefore, $\mathbf{W}$ is invertible.
\end{proof}

\section{Proof of Lemma \ref{lemma2}}\label{app2}
\begin{proof}
Note that $\mathbf{D}$ is a lower triangular matrix. To proof $\mathbf{D}$ is invertible, one way is to show that all its diagonal elements are nonzero. First, we know that $\mathbf{A}_1$ is a diagonal matrix with all diagonal entries greater than zero. Then, the remaining part is to show that diagonal elements of $\mathbf{T}_1$ is nonzero. Remind that $\mathbf{T}_1=\mathbf{B}_3\mathbf{B}_1^{-1}\mathbf{B}_2-\mathbf{B}_4$. For convenience, denote $\mathbf{B}_1^{-1}=[s_{ij}^1]_{i,j\in\mathcal{N}_d}$. Then, based on Lemma 1, we have $s_{ij}^1\leq 0,\ \forall i,j\in\mathcal{N}_d$. Since $\mathbf{B}$ is symmetric, then $\mathbf{B}_3^T=\mathbf{B}_2$. In addition, the entries of $\mathbf{B}_2$ and $\mathbf{B}_3$ are non-negative from definition. After algebraic calculations, we obtain that matrix $\mathbf{B}_3\mathbf{B}_1^{-1}\mathbf{B}_2\in\mathbb{R}^{N_g\times N_g}$ has non-negative diagonal elements. Since the diagonal entries of $\mathbf{B}_4$ are smaller than zero, then, $\mathbf{T}_1$ has all positive diagonal entries, and thus $\mathbf{D}$ is invertible.
\end{proof}

\bibliographystyle{IEEEtran}
\bibliography{IEEEabrv,microgrids_stakelberg_revised}

\end{document}